\documentclass[journal]{IEEEtran}

\usepackage{amsmath,amsfonts,amssymb,amsthm,dsfont}
\usepackage{mathtools}
\usepackage{tabularx}
\usepackage{caption}
\usepackage{subcaption}
\usepackage{graphicx}
\usepackage{algorithmic}
\usepackage{array}
\usepackage{textcomp}
\usepackage{stfloats}
\usepackage{url}
\usepackage{verbatim}
\usepackage[super]{nth}
\usepackage{xcolor}
\usepackage{balance}
\usepackage[hidelinks]{hyperref}
\usepackage{cite}

\newcommand{\Av}[2] {\left \langle  {#2} \right \rangle_{#1}}
\newcommand{\E} {{\mathbb E}}
\newcommand{\Z} {{\mathbb Z}}
\newcommand{\R} {{\mathbb R}}

\newcommand{\om} {{\omega}}
\newcommand{\eps} {\epsilon}
\newcommand{\mness} {\mathcal M}
\newcommand{\pa}[1] {#1_{\text{past}}}
\newcommand{\fu}[1] {#1_{\text{future}}}
\newcommand{\sple}[1] {\widetilde{#1}} 
\newcommand{\sha} {H_\eta(\pa{\sple{x}})}

\newtheoremstyle{break}
  {\topsep}{\topsep}%
  {\itshape}{}%
  {\bfseries}{}%
  {\newline}{}%
\theoremstyle{break}
\newtheorem{theorem}{Theorem}

\newtheorem{lemma}[theorem]{Lemma}
\newtheorem{proposition}[theorem]{Proposition}

\begin{document}

\title{Path Shadowing Monte-Carlo}

\author{
Rudy Morel$^a$, 
St\'ephane Mallat$^{a,b}$ \& 
Jean-Philippe Bouchaud$^c$
\\
\vspace{0.2cm}
\textit{
$^a$\'Ecole Normale Sup\'erieure, Paris, France
\\
$^b$Coll\`ege de France, Paris, France
\\
$^c$Capital Fund Management, Paris, France}
}

\maketitle

\begin{abstract}
We introduce a \textit{Path Shadowing Monte-Carlo} method, which 
provides prediction of future paths, given any generative 
model.
At any given date, it averages future quantities over generated price paths whose past history matches, or `shadows', the actual (observed) history. 
We test our approach using paths generated from a maximum entropy model of financial prices, based on a recently proposed multi-scale analogue of the standard skewness and kurtosis called `Scattering Spectra'~\cite{morel2022scale}. 
This model promotes diversity of generated paths while reproducing the main statistical properties of financial prices, including  stylized facts on volatility roughness. Our method yields state-of-the-art predictions for future realized volatility and allows one to determine conditional option smiles for the S\&P500 that outperform both the current version of the Path-Dependent Volatility model and the option market itself
\footnote{This work is supported by the PRAIRIE 3IA Institute of the French ANR-19-P3IA-0001 program and the ENS-CFM models and data science chair.}.
\end{abstract}


\section{Introduction}

Modelling future price scenarios is crucial for risk control, for pricing and hedging contingent claims (like options), and, possibly, for detecting arbitrage opportunities. Recently, machine learning auto-regressive models~\cite{oord2016wavenet,vaswani2017attention,wen2022transformers} manage to learn from data the distribution $p(x|\pa{x})$ of log-prices $x$ conditioned on past history $\pa{x}$. When trained with a prediction loss, such models generally achieve excellent prediction results. However, their training requires very large amount of data which are usually not available for financial prices.

On the other hand, low-parameterized generative models, i.e. models $p_\theta(x)$ of $p(x)$ with few parameters $\theta$, have been extensively studied in the financial literature~\cite{Heston1993ACS,bacry2013log,gatheral2018volatility,wu2022sfbm,guyon2021volatility}. However, two main challenges come to the fore. 
First, these models may not reproduce some important statistics of real financial prices due to oversimplified or flawed assumptions, or due to the fact that they are calibrated on external data such as observed option smiles. Second, it may not be straightforward to condition these models on the realized past at a specific date, in other words, obtaining a model of $p(x|\pa{x})$. Whereas conditioning is eased by considering Markovian models with a small number of factors~\cite{Heston1993ACS,guyon2021volatility}, such a strong assumption is often much too simplistic. 

\begin{figure}[!t]
\centering
\includegraphics[width=0.9\linewidth]{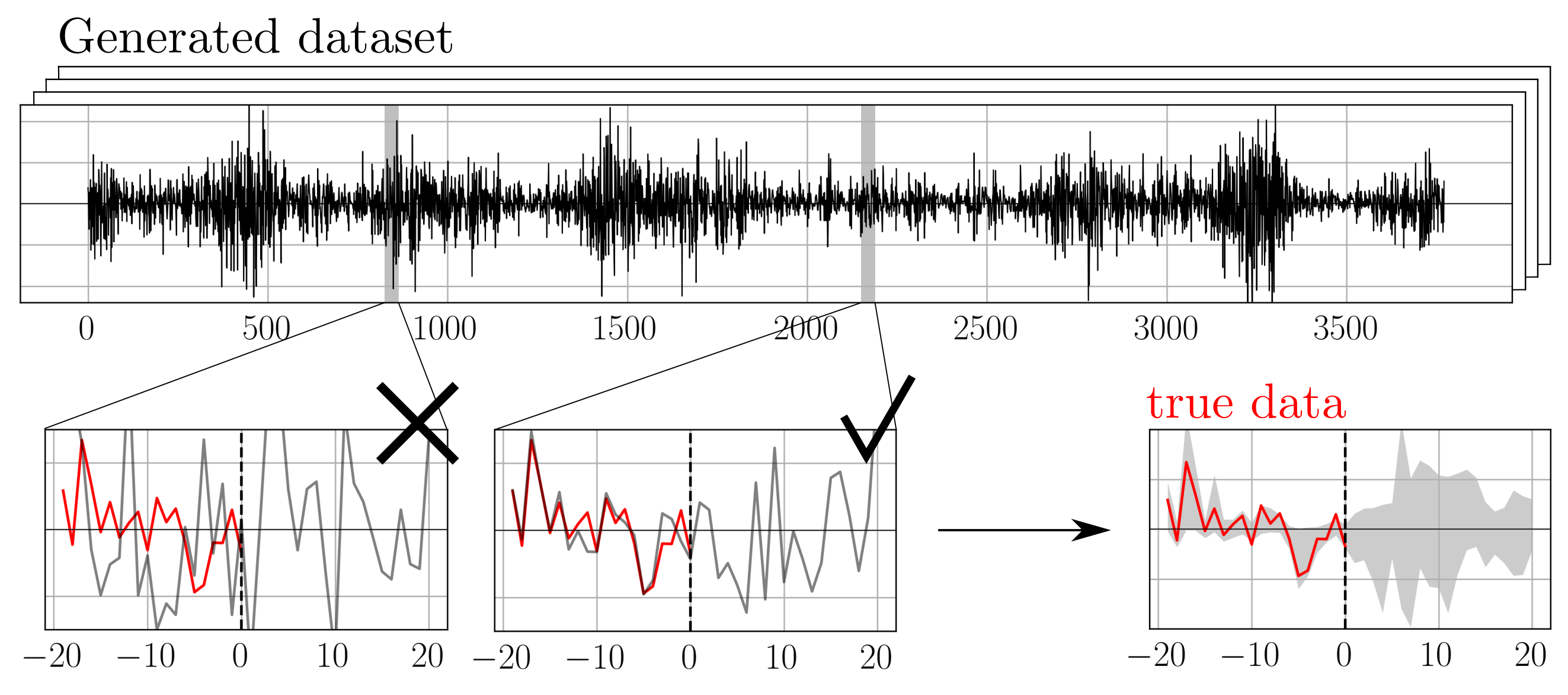}
\caption{Path shadowing. Given current past history $\pa{\sple{x}}$ (red), we scan for paths $x$ (gray) in a generated dataset whose past history satisfies $\pa{x}\approx\pa{\sple{x}}$. 
Such paths $x$ are said to \textit{shadow} $\pa{\sple{x}}$, they provide insights on the future. 
Predictions are obtained through Monte-Carlo on such shadowing paths.
}
\label{fig:shadowing}
\end{figure}

In this paper, we attempt to address both challenges.
Our main contribution is to introduce a new method, that we call \textit{Path Shadowing Monte-Carlo} (PS-MC), which can be used within any generative model of $p(x)$ to yield a model of $p(x|\pa{x})$.
Our approach for modelling the distribution $p(x)$, summarized in section \ref{sec:a-statistical-model}, is to define a minimal set of statistics describing financial prices that should be reproduced by the generating process. 
Such statistics should be well-estimated on limited data, while specifying `relevant' properties of the process, in a sense made precise below.
This bias-variance tradeoff was addressed in our previous work~\cite{morel2022scale} in the general case of multi-scale processes, where it was shown that a good description of financial prices can be achieved by multi-scale analogues of the classical skewness and kurtosis, called \textit{Scattering Spectra}, which is motivated in section \ref{sec:a-statistical-model}. 

A model based on these statistics captures all important stylized facts such as fat-tail distributions, intermittency, leverage effect and the `Zumbach effect'~\cite{morel2022scale}. Section \ref{sec:avg-smile} characterizes the average shape of option smiles generated by our Scattering Spectra model and shows that it accounts in particular for the power-law behaviour of the at-the-money skew as a function of maturity, which were recently argued to be a specific feature of rough volatility models~\cite{gatheral2018volatility, fukasawa2017short}. The power-law behaviour of the kurtosis, first reported in \cite{potters1998financial}, is also remarkably well accounted for. 

`Path shadowing' is presented in section \ref{sec:path-shadowing}. It consists in softening the conditioning on a given past history $\pa{x}$. In a nutshell, 
it amounts to scanning a large generated dataset, in search of paths whose history closely `shadows' the actual history (see Fig.~\ref{fig:shadowing} for an illustration).
A Path Shadowing Monte-Carlo method then averages the quantity of interest over the future of such matching paths. 
The term `shadowing' is freely inspired by the shadowing principle in chaotic dynamical systems~\cite{anosov1969geodesic,sinai1972gibbs,bowen1975omega,hammel1987numerical}. Intuitively, it states that a path which is uniformly close to a true orbit will stay close (shadow) a true path for all time.

This method can effectively be seen as a kernel method, with a causal path embedding to reduce the dimensionality of recent past 
history.
Unlike other recent kernel methods, such as signature kernels~\cite{salvi2021higher,alden2022model} that rely on a low-parametric model for $p(x|\pa{x})$, Path Shadowing Monte-Carlo relies solely on a model of $p(x)$.
It thus circumvents the exact conditioning of a generative model to a given past history $\pa{x}$. 
Its performance depends directly on the accuracy of this generative model and its ability to produce a variety of paths with correct statistical dependencies. Section \ref{subsec:volatility-prediction} shows that when performed with our maximum entropy Scattering Spectra model of financial prices, PS-MC yields state-of-the-art volatility prediction.

Section \ref{sec:option-pricing} uses Path Shadowing Monte-Carlo for obtaining {\it conditional} option smiles (i.e. option prices at a given date) through Hedged Monte-Carlo with shadowing paths. 
By construction, such smiles depend only on the log-price process distribution $p(x)$ and provide a counterpart to smiles obtained from option market data.
A `trading game' then allows us to show  that our option smiles correctly anticipate non-trivial future price movements, and compares favourably with state of the art models such as the Path-Dependent Volatility model (PDV) of ref.~\cite{guyon2021volatility}. Codes for both our generative model and Path Shadowing Monte-Carlo are available at \url{https://github.com/RudyMorel/shadowing}.

\section{A Multi-Scale Statistical Model for Financial Prices}
\label{sec:a-statistical-model}

Statistical models of financial prices aim at reproducing statistics of the price process only. Price time series exhibit numerous non-Gaussian features, which are difficult to capture within standard low-parametric models, whose number of parameters have been incrementally increased in the literature over the past decades, see e.g.~\cite{Heston1993ACS,bacry2013log,gatheral2018volatility,delemotte2023twoHurstmodel, guyon2021volatility}.
An alternative route is to define a set of characteristic statistics of (log-)prices and impose that they should be accurately reproduced by the model. We denote as $\Phi(x)$ such statistics, for example the empirical mean and variance of log-returns $\Phi(x)=(\Av{t}{\delta x(t)},\Av{t}{|\delta x(t)|^2})$ where, throughout this paper, $\Av{t}{\ldots}$ denotes empirical averages over time $t$. Section \ref{subsec:maximum-entropy} presents maximum entropy models that allow defining models from a given vector of statistics $\Phi(x)$. In the simple case of mean and variance, the maximum entropy model coincides with the Gaussian random walk. 

The set $\Phi(x)$ must be chosen carefully and is governed by a bias-variance tradeoff. It should contain enough relevant statistics of prices for the model to be realistic and accurate -- reducing the model bias. 
However, such statistics must be well estimated on the single available historical realization of $x$ -- reducing the model variance. The construction of a set $\Phi$ that meets these requirements, called the \textit{Scattering Spectra} (SS), was proposed in~\cite{morel2022scale} in the general case of multi-scale processes, which fortunately includes financial data \cite{mandelbrot1997multifractal,bacry2013log, borland2005dynamics}. 

We show in section \ref{subsec:scattering-spectra} that such statistics correspond to low moment multi-scale analogue of the classical skewness and kurtosis of log-returns. We show that even the most recent low-dimensional parametric models fail to accurately account for these statistics. Such discrepancies turn out to be highly relevant when one wants to predict future realized volatility and option smiles, and highlights  the limitations of traditional models, which our approach allows one to overcome.  

In~\cite{morel2022scale}, we have shown that a Scattering Spectra model properly captures the main properties of financial log-returns, in particular of the S\&P500 (the US major stock index). 
In the following, we show that it also quantitatively reproduces the average behavior of option smiles of different maturities, in particular the maturity-dependent skewness that reflects volatility roughness~\cite{fukasawa2017short} and the so-called skew-stickiness ratio~\cite{bergomi2009smile, vargas2015skew}.

\subsection{Maximum Entropy Models}
\label{subsec:maximum-entropy}

We denote as $\sple{x}\in\R^N$ the observed historical realization of log-prices over $N$ days. 
Given a vector of $d$ statistics $\Phi(\sple{x})\in\R^d$ estimated on $\sple{x}$, a maximum entropy model $p_\theta$ with moment constraint $\E_{p_\theta}\{\Phi(x)\}=\Phi(\sple{x})$, if it exists, has an exponential probability distribution~\cite{jaynes_1957}
\begin{equation}
\label{eq:exp-model}
p_\theta (x) = Z^{-1}_\theta e^{- \langle \theta , \Phi(x) \rangle}.
\end{equation}
for certain $\theta\in\R^d$.
Estimating the parameters $\theta$ of model (\ref{eq:exp-model}) is computationally expensive, in particular when the number of statistics $d$ is large~\cite{rolf1998microsampling,betancourt2017conceptual}. 
To avoid this issue, we consider in this paper microcanonical maximum entropy models which approximate model (\ref{eq:exp-model}). These models, together with a sampling algorithm are described in Appendix~\ref{app:sampling}.

Maximum entropy models depend only on the vector of statistics $\Phi(x)$. The model accuracy can be improved by enriching the set $\Phi(x)$. However, we must take into account the problem of estimating $\Phi(x)$ from the single realization of the process $\sple{x}$.
The SS model imposes $\E_{p_\theta}\{\Phi(x)\} = \Phi(\sple{x})$, thus for $p_\theta$ to be a good approximation of the true distribution $p$, one 
needs
$\Phi(\sple{x})$ to be close to 
the true $\E_p\{\Phi(x)\}$. This amounts to having low-variance statistics $\Phi$. A good choice of $\Phi$ is presented in the next section.

\subsection{The Scattering Spectra (SS)}
\label{subsec:scattering-spectra}

A standard way of characterizing the price process is through their trend, volatility, skewness and kurtosis. These are obtained from moments of order $1$, $2$, $3$ and $4$ on log-returns
\begin{equation}
\nonumber
\label{eq:standard-marginal-moments}
\E\{\delta x(t)\} ~,~
\E\{\delta x(t)^2\} ~,~
\E\{\delta x(t)^3\} ~,~
\E\{\delta x(t)^4\}
\end{equation}
However such moments do not characterize the time-structure of log-returns, but rather their one-point distribution.
One could consider the same moments on multi-scale increments
\begin{equation}
\label{eq:multi-scale-increments}
\delta_\ell x(t) = x(t) - x(t-\ell) 
\end{equation}
for different lags $\ell$, but we still obtain a poor description of $x$. For example, these moments do not pick up time-asymmetry, since changing $\delta x(t)$ into $\delta x(-t)$ leaves these moments unchanged. Another disadvantage of multi-scale increments (\ref{eq:multi-scale-increments}) is that they exhibit 
as many scales $1\leq\ell\leq N$ as the 
number of days $N$,
which seems redundant, specially in view of the known scale-invariant properties of $x$. 

The \textit{Scattering Spectra} introduced in~\cite{morel2022scale} capture the main non-Gaussian properties of financial prices: fat tailed log-return distributions, sign-asymmetry, time-asymmetry, volatility clustering and volatility roughness. It consists of $d=\mathcal{O}(\log_2^3 N)$ statistics only, that are low-order moments (order $1$ and $2$ only) and can thus be accurately estimated on the historical realization $\sple{x}$ of size $N$. 

We present here the main steps for building such $\Phi$ and we refer the reader to~\cite{morel2022scale} for more details about the construction.

\noindent\textbf{Step 1. Wavelet increments.}

Log-prices variation have interesting structure at all scales. However, it is not necessary to consider all scales $\ell$ in (\ref{eq:multi-scale-increments}) to characterize them efficiently. 
Standard increments $\delta_\ell x(t)$ (\ref{eq:multi-scale-increments}) are obtained by convolution of $x$ with the filter $g_\ell = \delta_0-\delta_\ell$. 
Wavelet increments replace $g_\ell$ by wavelet filters $\psi_j$ obtained by dilation of a regular mother wavelet $\psi$
\begin{equation}
\label{eq:wavelet-increments}
W_j x(t) = x\star\psi_j(t)~~\mbox{where}~~\psi_j(t) = 2^{-j}\psi(2^{-j}t).
\end{equation}
The mother wavelet $\psi$ has a zero average $\int \psi(t) {\rm d}t = 0$ and its Fourier transform 
$\widehat{\psi}(\om) = \int \psi(t)\, e^{-i \om t}\, {\rm d}t$, which is real, is mostly concentrated at frequencies $\om \in [{\pi}, {2}\pi]$. 
All numerical calculations in this paper are performed with a complex Battle-Lemarié wavelet~\cite{battle1987block,lemarie1988ondelettes}. Fig.~\ref{fig:wavelet} shows the real and imaginary parts of $\psi$ as well as its Fourier transform. We refer the reader to Appendix \ref{app:wavelet} for more properties.

Analogous to (\ref{eq:multi-scale-increments}), wavelet increments (\ref{eq:wavelet-increments}) can be seen as multi-scale increments at scales $\ell=2^j$ with $j=1,\ldots,J$.
However, scales are now defined as bins of frequencies $[2^{-j}\pi, 2^{-j+1}\pi]$ corresponding to the supports of wavelet filters $\psi_j$. The largest scale $2^J$ is chosen to be smaller than the size $N$ of $\sple{x}$. This yields at most $\log_2 N$ scales instead of $N$ lags $\ell$.

Histograms of generalized increments $W_j x$ can be constrained by order 1 and order 2 moments $\E\{|W_j x(t)|\}$, $\E\{|W_j x(t)|^2\}$ which are estimated through  empirical averages. The quantity
\begin{equation}
\label{eq:phi1}
\Phi_1(x)[j] =  \frac{\Av{t}{|W_j x(t)|}^2}{\Av{t}{|W_j x(t)|^2}} 
\end{equation}
is a low-moment measure of kurtosis. Compared to its order $4$ counterpart, it is less sensitive to large values. The more peaked at zero the distribution, the smaller the value of $\Phi_1(x)$ and the higher the kurtosis~\cite{de2012avgsmile}. The order $2$ moment is
\begin{equation}
\label{eq:phi2}
\Phi_2(x)[j] = \Av{t}{|W_j x(t)|^2}
\end{equation}
and quantifies the average volatility at scale $2^j$ on the period.

\noindent\textbf{Step 2. Time-scale dependencies.}

Multi-scale increments $W_jx(t)$ are indexed by time $t$ and scale $2^j$. 
Such map exhibits dependencies across time and scales that are crucial to characterize the distribution of financial prices.
For example, volatility clustering is attested by the fact that $W_jx(t)$ has long-range time correlations. 
Beyond this well-known stylized fact, the authors of~\cite{morel2022scale} have shown that scale dependencies are crucial to fully characterize the non-Gaussian nature of time series. Natural descriptors for such scale dependencies are  order $2$, $3$ and $4$ moments
\begin{equation*}
\E\{ Wx\,Wx^* \} ~,~ 
\E\{ Wx\,|Wx|^2 \} ~,~ 
\E\{ |Wx|^2\,|Wx|^2 \}
\end{equation*}
where the products are taken across times $t,t'$ and scales $j,j'$. In practice, estimating order $3$ and order $4$ moments is very difficult because of the variance induced by large events.
In order to circumvent this problem, we replace $|Wx|^2$ by $|Wx|$ and define the following non-linear correlations of wavelet increments
\begin{equation}
\label{eq:multi-scale-1234}
\E\{ Wx\,Wx^*\} ~,~ 
\E\{ Wx\,|Wx|\} ~,~ 
\E\{ |Wx|\,|Wx|\}
\end{equation}
Owing to the compression properties of wavelets, the first matrix $\E\{ Wx\,Wx^*\}$ is quasi-diagonal and its diagonal coefficients are already estimated by (\ref{eq:phi2}), see \cite{morel2022scale}.

\begin{figure}
\centering
\begin{subfigure}[b]{\linewidth}
    \centering
    \includegraphics[width=\textwidth]{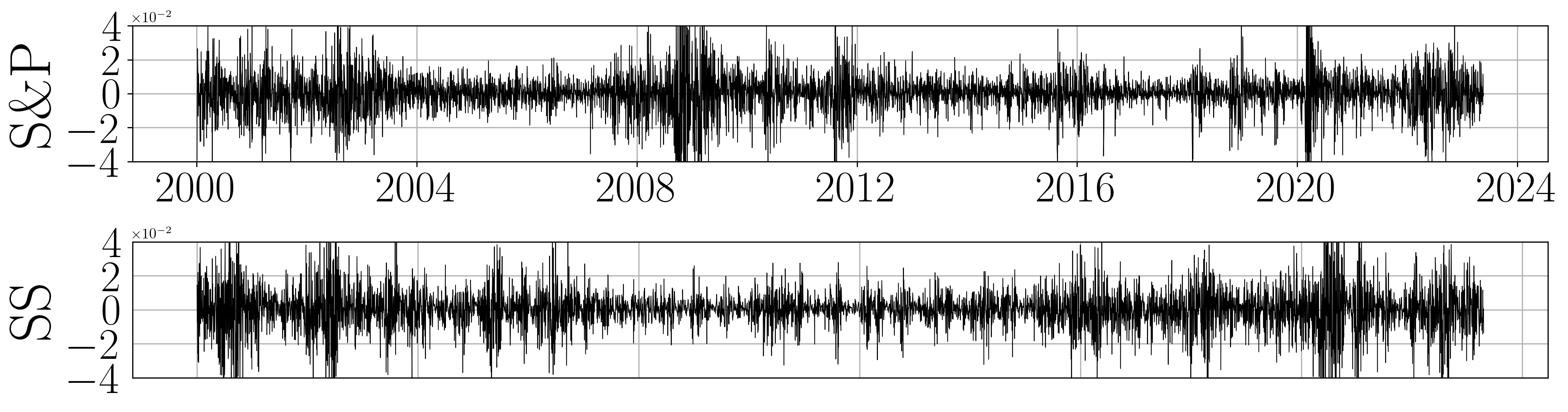}
\end{subfigure}
\vskip\baselineskip
\begin{subfigure}[b]{0.32\linewidth}
    \centering
    \includegraphics[width=\textwidth]{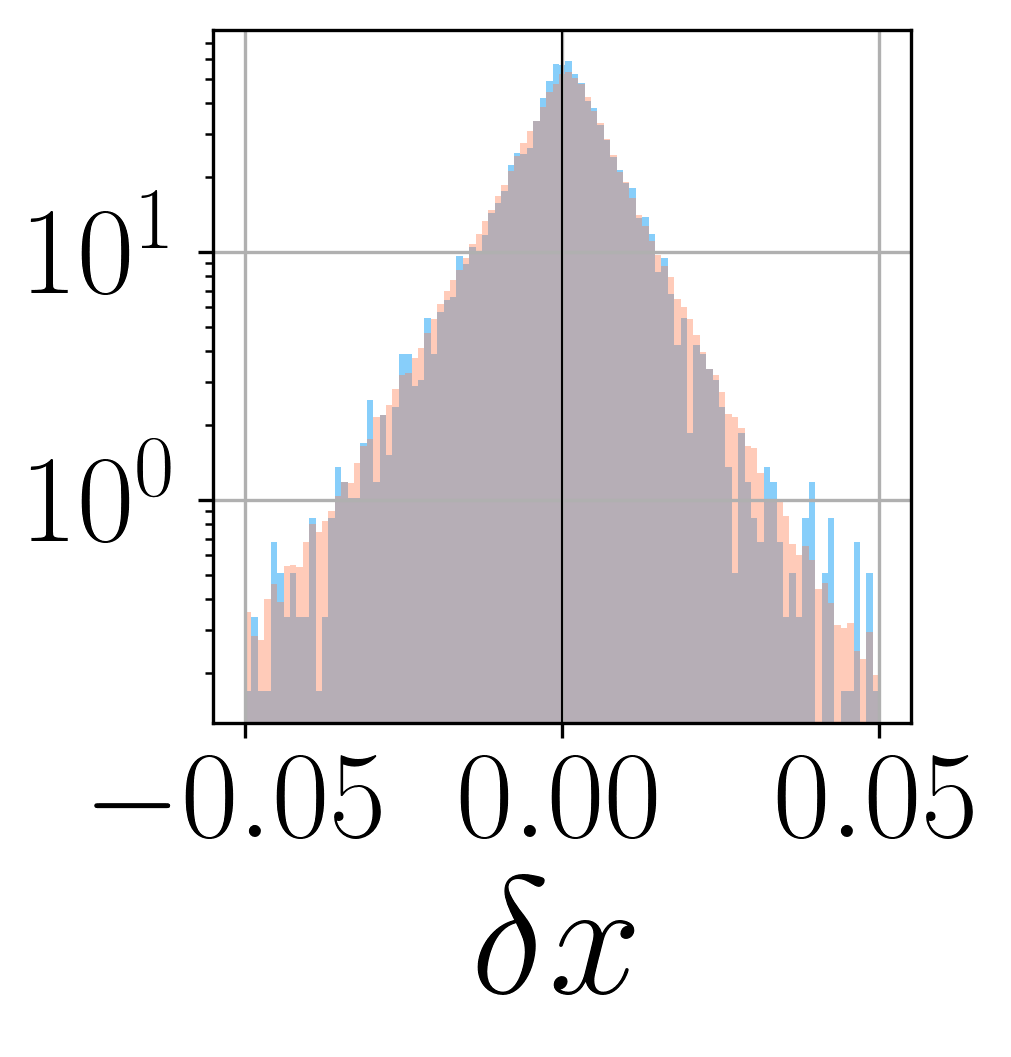}
    \caption{Histogram}
\end{subfigure}
\begin{subfigure}[b]{0.32\linewidth}  
    \centering 
    \includegraphics[width=\textwidth]{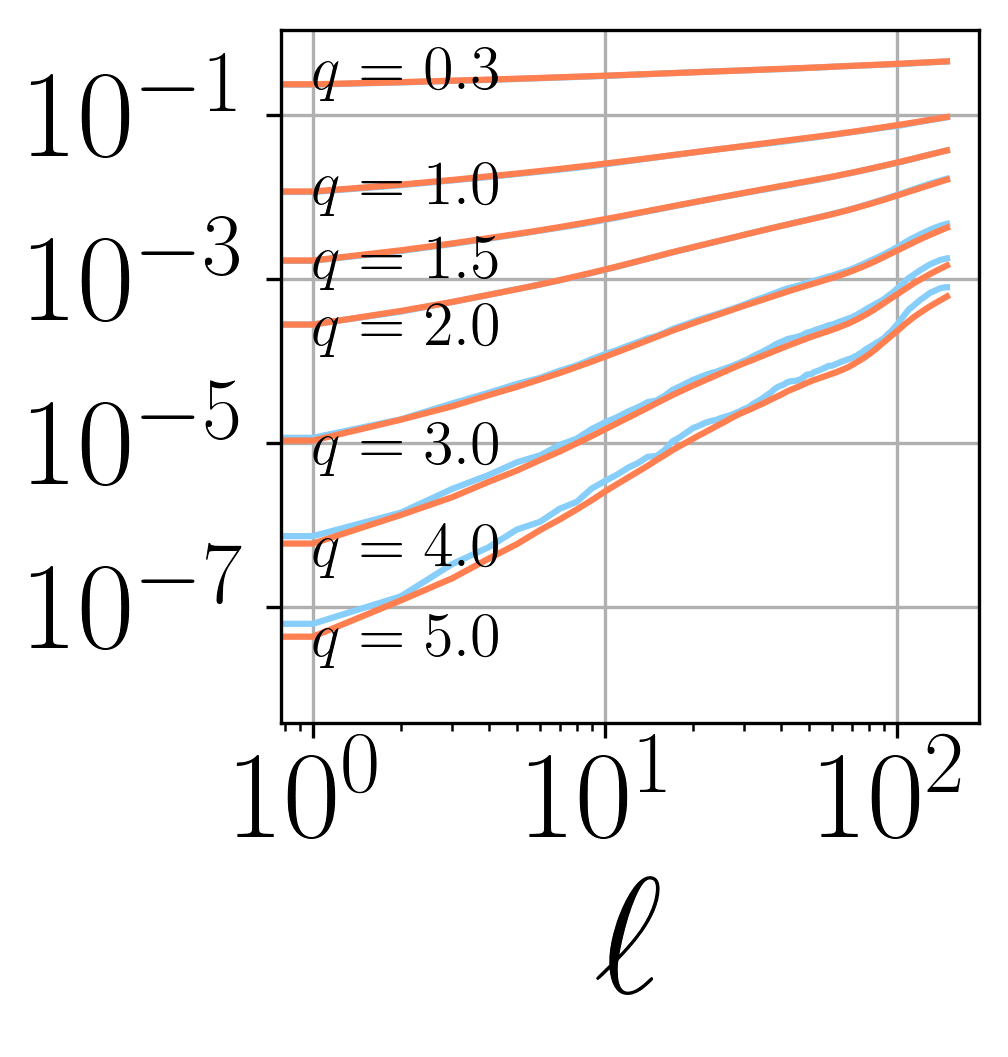}
    \caption{Struct. functions}
\end{subfigure}
\begin{subfigure}[b]{0.32\linewidth}   
    \centering 
    \includegraphics[width=0.96\textwidth]
    {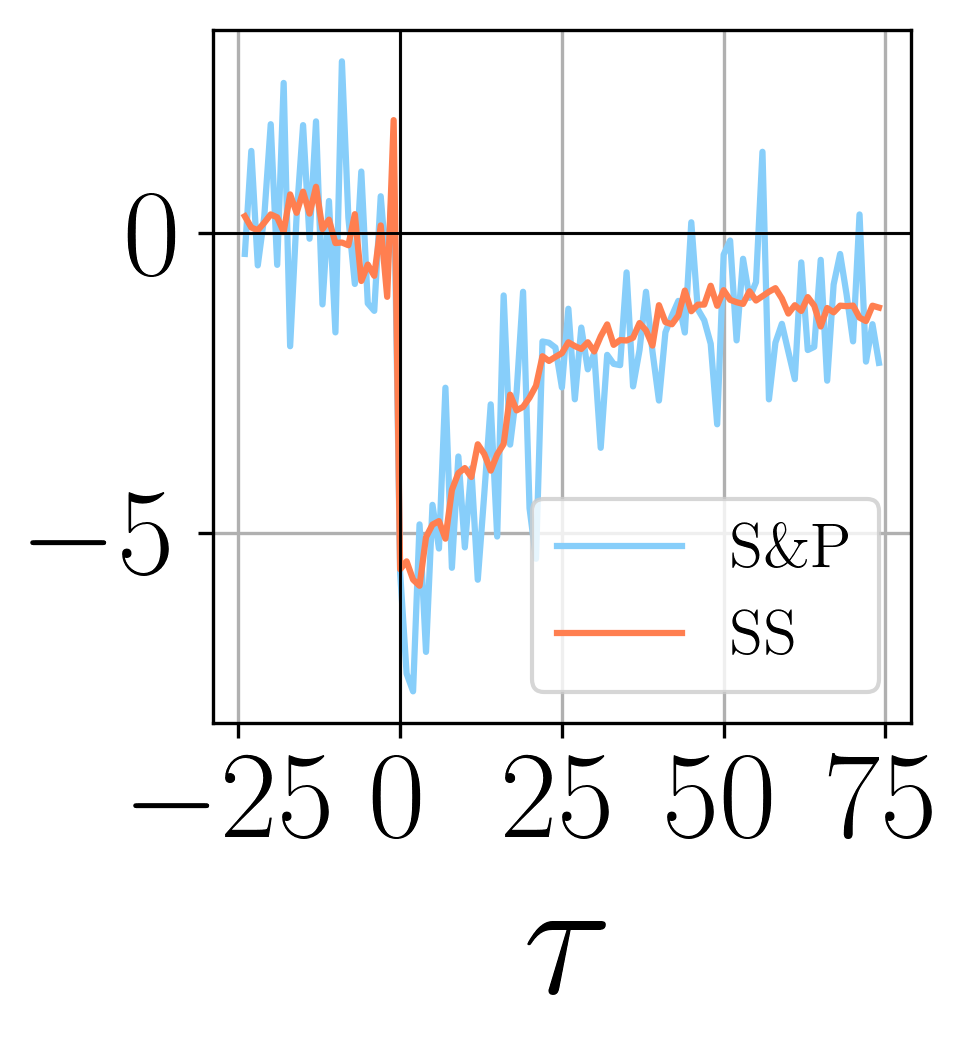}
    \caption{Leverage}
\end{subfigure}
\caption{Standard statistics of log-returns in the Scattering Spectra (SS) model (orange) compared to S\&P observed data (blue). Top graphs: time series of the S\&P and generated by the model. Bottom graphs:
(a) Histogram of daily log-returns $\delta x$. 
(b) Structure functions $\Av{t}{|\delta_\ell x(t)|^q}$. 
(c) Leverage correlation $\Av{t}{\delta x(t-\tau)|\delta x(t)|^2}$ on normalized increments. 
}
\label{fig:stat-our-model}
\end{figure}

\noindent\textbf{Step 3. Low-moment multi-scale skewness and kurtosis.}

Just like for standard skewness and kurtosis that are normalized moments, we normalize the second and third matrices $\E\{ Wx\,|Wx|\}$ and $\E\{ |Wx|\,|Wx|\}$ in (\ref{eq:multi-scale-1234}) by $\E\{|Wx|^2\}$.
One can show that the only non-negligible coefficients in the third matrix are obtained for $t=t'$ and $j\geq j'$, and are estimated through
\begin{equation}
\label{eq:phi3}
\Phi_3(x)[j,j'] = \frac{\Av{t}{W_jx(t)\,|W_{j'}x(t)|}}{\Av{t}{|W_j x(t)|^2}^{\frac{1}{2}}\Av{t}{|W_{j'} x(t)|^2}^{\frac{1}{2}}}.
\end{equation}
These are analogous to the standard low-moment skewness $\E\{Y|Y|\}$ of a normalized random variable $Y$ with $\E\{Y^2\}=1$.
Other than sign-asymmetry, these complex coefficients also measure time-asymmetry through their phase. Indeed, if log-returns are time-reversible $\delta x(-t) \overset{d}{=}\delta x(t)$ then $\mathrm{Im}\,\Phi_3(x)=0$. One typical example is the leverage asymmetric correlation.

The fourth matrix $\E\{|Wx|\,|Wx|\}$ in (\ref{eq:multi-scale-1234}) contains kurtosis information. If $x$ is Gaussian, then for different scales $j \ne j'$ the Gaussian processes $W_j x$ and $W_{j'}x$ are decorrelated, thus independent. It follows that $\E\{|W_jx||W_{j'}x|\} = \E\{|W_jx|\}\E\{|W_{j'}x|\}$ and these coefficients boil down to the low-moment kurtosis (\ref{eq:phi1}). For the log-price process $x$, these coefficients capture long-range non-Gaussian correlation between volatility at different scales $j,j'$ and different times $t,t'$. 

However, the matrix $\E\{|Wx|\,|Wx|\}$ contains too many coefficients to be accurately estimated on a single realization $\sple{x}$. We again rely on compression properties of wavelets to approximate such matrix by cascading a second wavelet operator $W$, which yields a quasi-diagonal matrix $\E\{W|Wx|\,W|Wx|^*\}$ where we define generalized increments of volatility as
\begin{equation}
\nonumber
W_{j_2}|W_{j_1}x|(t) = |x\star\psi_{j_1}|\star\psi_{j_2}(t).
\end{equation}
The non-negligible diagonal coefficients are estimated through an empirical average which yields for $j_1\leq j'_1<j_2$
\begin{equation}
\label{eq:phi4}
\Phi_4(x)[j_1,j'_1,j_2] = \frac{\Av{t}{W_{j_2}|W_{j_1}x|(t)\,W_{j_2}|W_{j'_1}x|(t)}}{\Av{t}{|W_{j_1} x(t)|^2}^{\frac{1}{2}}\Av{t}{|W_{j'_1} x(t)|^2}^{\frac{1}{2}}}.
\end{equation}
These are analogous to the classical low-moment kurtosis. These complex coefficients also capture time-asymmetry through their complex phase. If the log-return process $\delta x$ is time-reversible then $\mathrm{Im}\,\Phi_4(x)=0$. 

We therefore define our Scattering Spectra $\Phi$ as the collection of (i) estimated average volatility (\ref{eq:phi2}), (ii) multi-scale skewness (\ref{eq:phi3}) and (iii) multi-scale kurtosis (\ref{eq:phi1},\ref{eq:phi4})
\begin{equation}
\label{eq:scattering-spectra}
\Phi(x) = \big (\Phi_1(x), \Phi_2(x), \Phi_3(x), \Phi_4(x)\big).
\end{equation}
In total, $\Phi$ consists of $\mathcal{O}(\log^3_2(N))$ order $1$ and order $2$ statistics for a trajectory of size $N$ and can be estimated with low-variance.

Note that $\Phi$ does not rely explicitly on the one-point distribution of increments $\delta_\ell x(t)$. 
Numerical experiments have shown that
slight discrepancies may appear, in particular in order $0$ moments $\mathbb{P}(\delta_\ell x(t)>0)$ which explicitly appear in low-moment smile expansions~\cite{de2012avgsmile}. We thus complement $\Phi_3(x)$ with the moments 
\[
\mathbb{P}(\delta_\ell x(t)>0)
\]
for $\ell=2^j,j=1,\ldots,J$, that are estimated through empirical averages $\Av{t}{\text{sigmoid}(\delta_\ell x(t))}$ where $\text{sigmoid}(x)=(1+e^{-x})^{-1}$.
This adds very few coefficients to our scattering spectra $\Phi(x)$. 

The Scattering Spectra (\ref{eq:scattering-spectra}) thus provide an enriched set of statistics that can be used to quantify model error and interpret any discrepancy. As an example, we revisit through this lens the, low-parametric, Path-Dependent Volatility (PDV) model introduced by Guyon \& Lekeufack~\cite{guyon2021volatility}. 
Although more parsimonious in terms of number of parameters, several stylized facts are in fact not accurately reproduced by such a model, see a more precise discussion in Appendix \ref{app:pdv}, Fig.~\ref{fig:stat-pdv}.

Based on the Scattering Spectra $\Phi$, we have at our disposal a statistical model of financial prices that can be used to generate faithful synthetic time series (see section \ref{subsec:maximum-entropy}). For the S\&P time series $\sple{x}$ of size $N=5827$ days, the Scattering Spectra model (SS model) contains $248 \approx N/20$ real coefficients, which is the dimension of $\Phi(x)$. Log-return trajectories $\delta x$ generated from the SS model are shown in Fig.~\ref{fig:stat-our-model}. Validation of the SS model can be achieved by measuring observables not included in our set $\Phi(x)$ and checking whether or not they are correctly reproduced. Standard statistics such as fat tails, volatility clustering, leverage effect and structure functions, were indeed shown to be captured by the model~\cite{morel2022scale}. These are reproduced in Fig.~\ref{fig:stat-our-model}. While $\Phi(x)$ is composed of order $1$ and order $2$ moments only, the SS model accurately accounts for up to order $5$ moments, which is quite remarkable. Another way to describe the multi-scale statistical properties of price time series is through maturity dependent option smiles, which we discuss in the next section.

\section{The Average Smile as an Alternative Statistical Characterization}
\label{sec:avg-smile}

In this section we validate the SS model by considering historical option pricing as an alternative, intuitive way to characterize the multi-scale, non-Gaussian statistics of price time series. The \textit{average smile} is the unconditional option smile obtained by pricing hedged options using all historical snippets of prices of length equal to the maturity of the option~\cite{potters2001HMC,de2012avgsmile}. Even if real option smiles must be conditioned on a specific past price path~\cite{guyon2021volatility} and are therefore almost never equal to the \textit{average smile}, its shape reveals some interesting, non-trivial properties of prices time series, such as volatility `roughness' (see below). 

Option pricing is performed through the Hedged Monte-Carlo method~\cite{potters2001HMC}, that converts historical probabilities (either real or synthetic) into `risk-neutral' ones. Options are hedged daily, with zero interest rate, on the 6000 price snippets of lengths 150 days available from 2000 to 2023 all rescaled such that the initial price is 100. The average implied volatilities $\sigma(T,K)$ are obtained from option prices $\mathcal{C}(T,K)$. 
Fig.~\ref{fig:avg-smile} compares, for different maturities $T$, the \textit{average smiles} using observed S\&P data and those generated with the SS model. We see that the model indeed reproduces the overall shape of the smile very well. 
Intuitively, the level of the average smile, its asymmetry, its concavity and its term structure are captured by $\Phi_2$ (\ref{eq:phi2}), $\Phi_3$ (\ref{eq:phi3}) and  $\Phi_1$ and $\Phi_4$ (\ref{eq:phi1},\ref{eq:phi4})
\footnote{Appendix \ref{app:smile-sensitivity} shows in more details the parameterization of the model by studying the sensitivity of the smile to the Scattering Spectra statistics $\Phi(x)$.}. 
We have also compared the S\&P average smiles 
with the recent Path-Dependent Volatility model of~\cite{guyon2021volatility},
calibrated such as to reproduce the same SS as best as possible -- see Appendix \ref{app:pdv} for more details. As a general comment, the PDV model underestimates the kurtosis of the process and correspondingly fail to capture accurately the right wing of the average smile, see Fig.~\ref{fig:smile-avg-pdv}. 

We now turn to a more refined analysis of the slope and curvature of these average smiles. 
We denote as $\sigma_{\text{ATM}}(T)=\sigma(T,100)$ the at-the-money volatility and 
\[ 
\mness:=\frac{\ln(\frac{{K}}{100})}{\sigma^{\text{ATM}}\sqrt{T}}
\]
the {\it rescaled} log-moneyness. The slope $\mathcal{S}_T$ and curvature $\kappa_T$ of a smile at maturity $T$ are defined by the order $2$ expansion around the moneyness $\mness=0$
\[
\sigma(\mness,T) := \sigma_{\mathrm{ATM}}(T)\bigg(1 + \mathcal{S}_T  \mness + \kappa_T\mness^2 + o(\mness^2)\bigg)
\]
In the literature, it is customary to define the ATM skew $\mathrm{Skew}_T$ as the slope of the smile as a function of \textit{unscaled} log-moneyness, i.e. $\mathrm{Skew}_T := \mathcal{S}_T/\sqrt{T}$. For most stochastic volatility models, such skew is found to be regular when $T \to 0$, whereas rough volatility models predict a singular behavior $\mathrm{Skew}_T \propto T^{H - 1/2}$ where $H$ is the Hurst exponent of volatility, argued to be small, $H \approx 0.1$~\cite{gatheral2018volatility,fukasawa2021volatility,bayer2016pricing}.

\begin{figure}[h]
\centering
\begin{subfigure}[b]{0.3\linewidth}
    \centering
    \includegraphics[width=\textwidth]{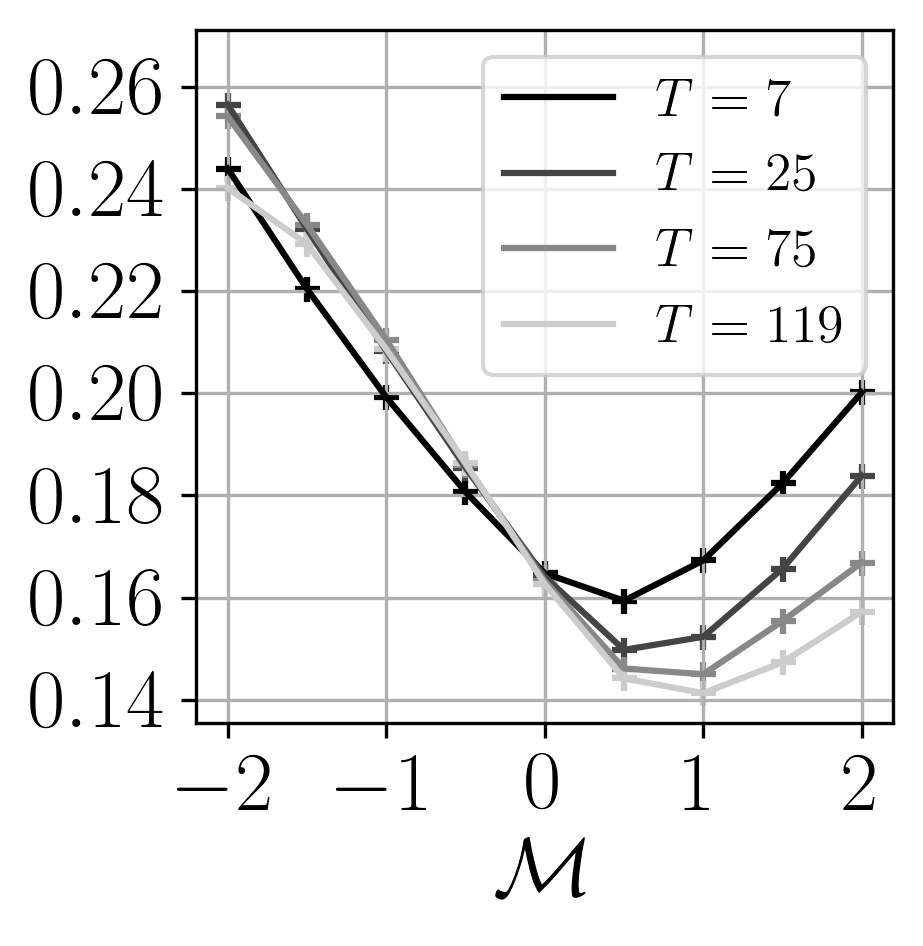}
    \caption{Smile (S\&P)}    
\end{subfigure}
\begin{subfigure}[b]{0.3\linewidth}  
    \centering 
    \includegraphics[width=\textwidth]{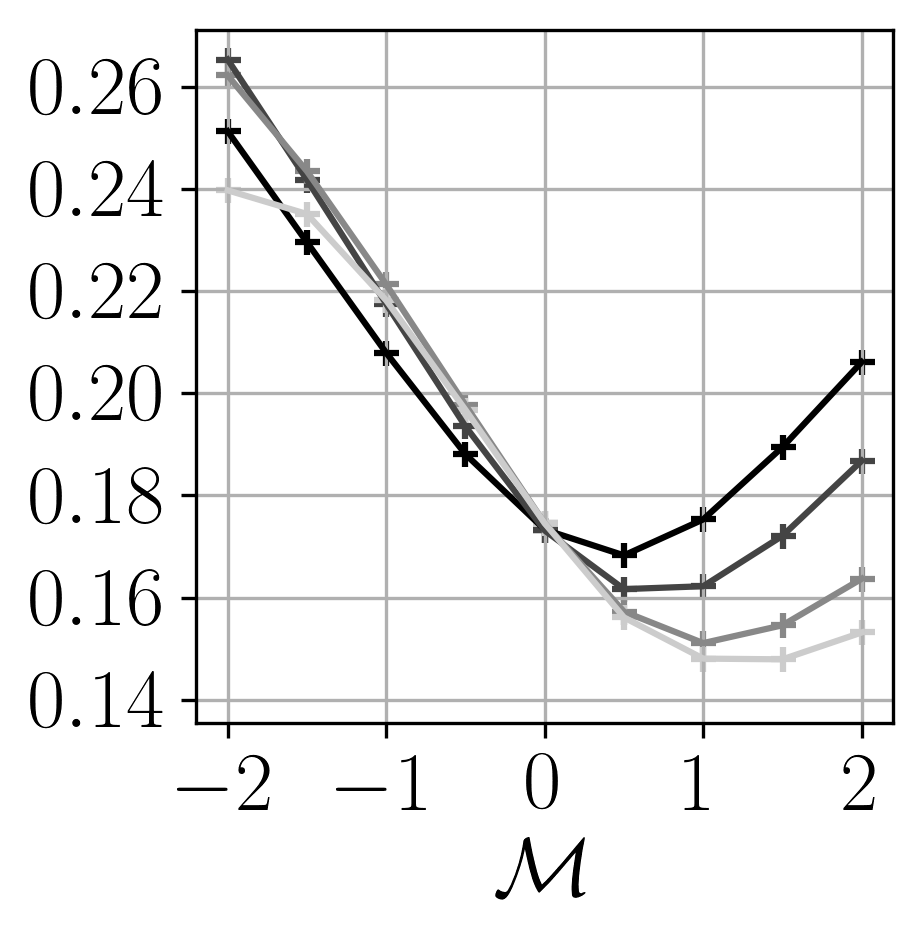}
    \caption{Smile (SS)}
    \label{fig:smile-avg-ss}
\end{subfigure}
\begin{subfigure}[b]{0.3\linewidth}  
    \centering 
    \includegraphics[width=\textwidth]{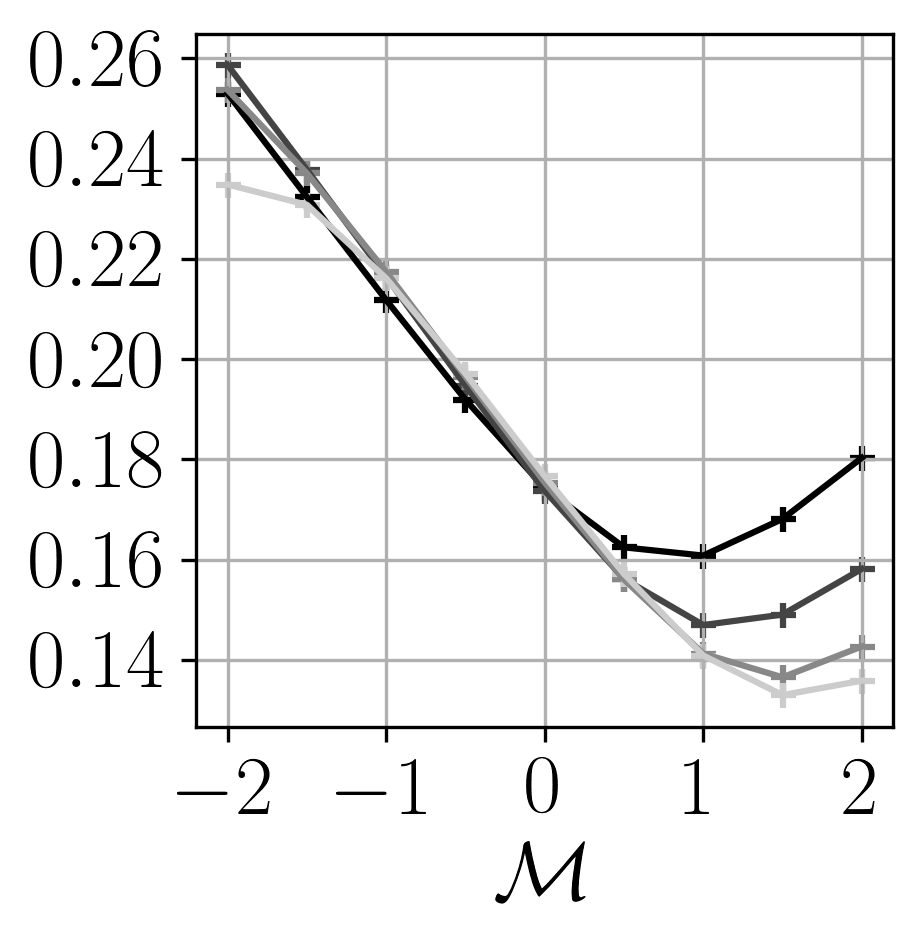}
    \caption{Smile (PDV)}
    \label{fig:smile-avg-pdv}
\end{subfigure}
\vskip\baselineskip
\begin{subfigure}[b]{0.31\linewidth}   
    \centering 
    \includegraphics[width=\textwidth]{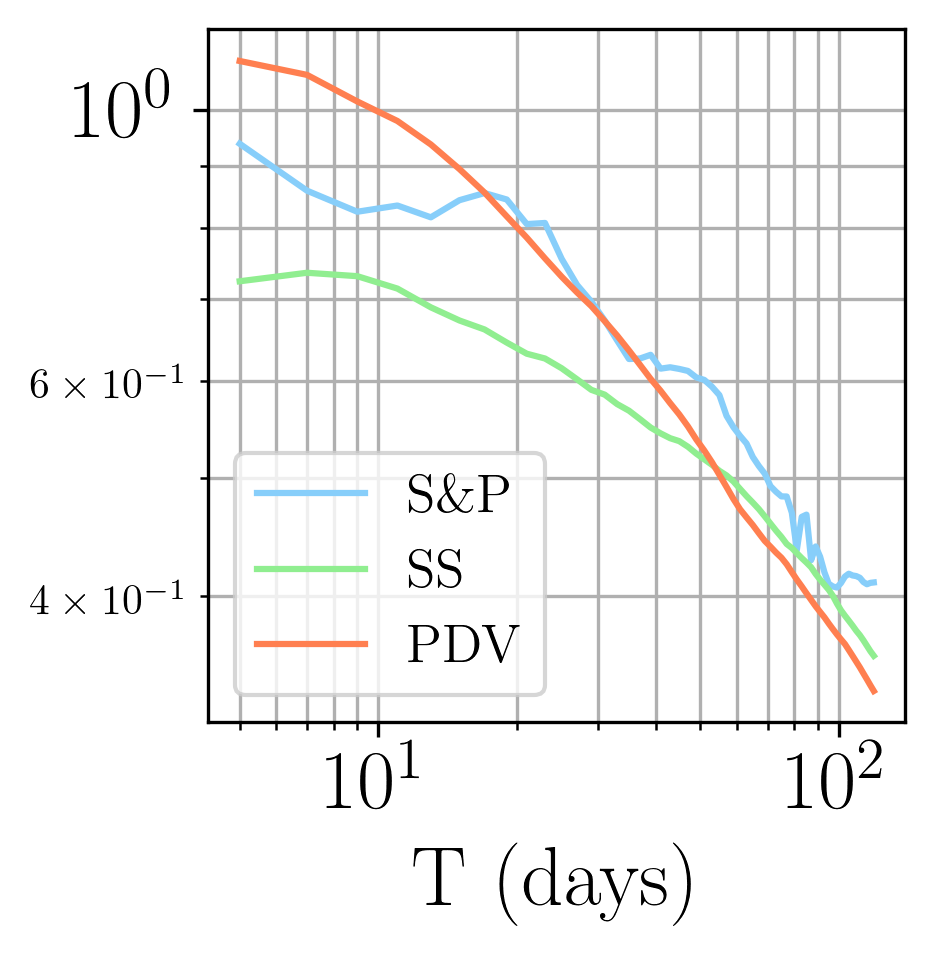}
    \caption{ATM skew (abs)}
    \label{fig:atm-skew}
\end{subfigure}
\begin{subfigure}[b]{0.31\linewidth}   
    \centering 
    \includegraphics[width=\textwidth]{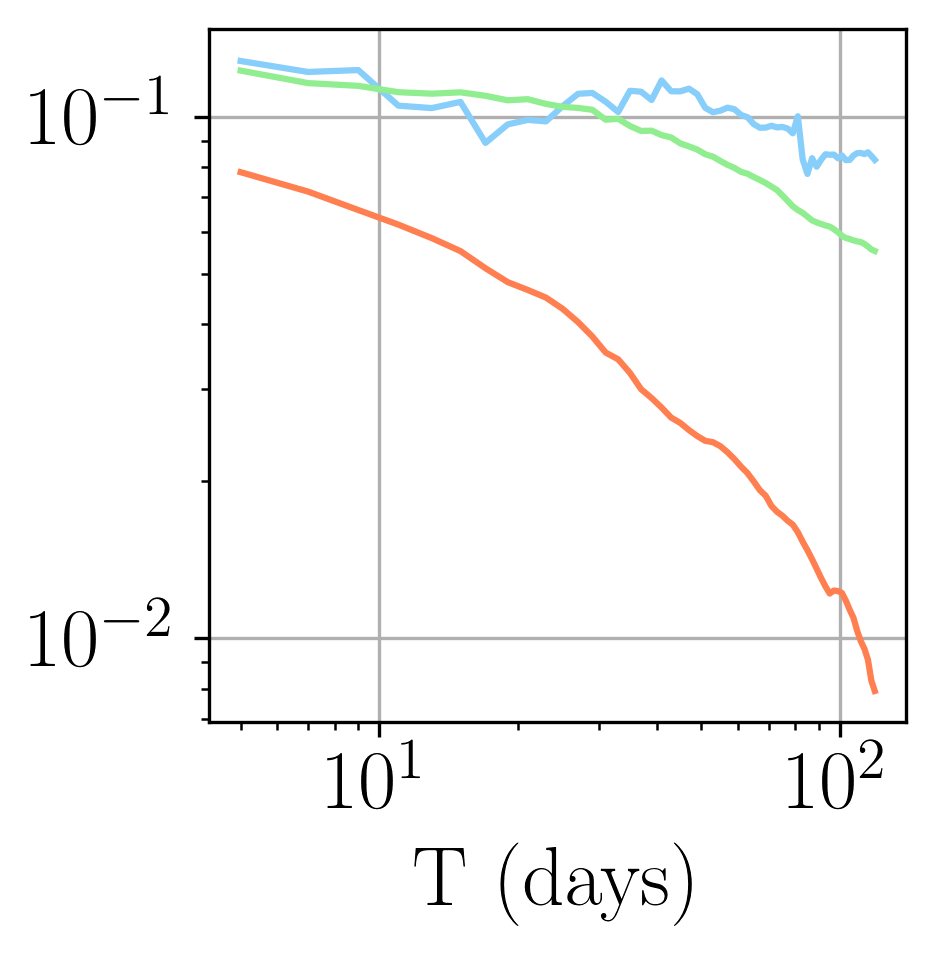}
    \caption{ATM kurtosis}
    \label{fig:atm-kurtosis}  
\end{subfigure}
\begin{subfigure}[b]{0.3\linewidth}   
    \centering 
    \includegraphics[width=\textwidth]{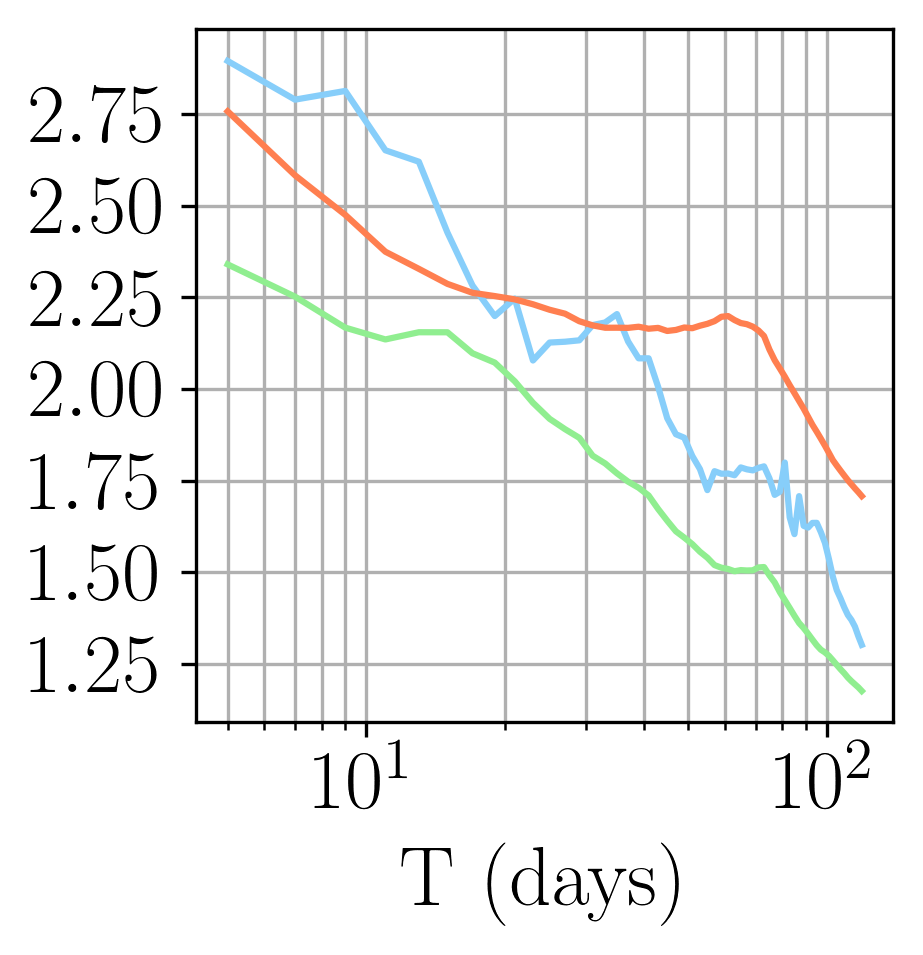}
    \caption{Skew Stickiness}
    \label{fig:sksratio}    
\end{subfigure}
\caption{Average option smiles estimated using S\&P price data (a), in our Scattering Spectra (SS) model (b) and in the Path-Dependent Volatility (PDV) model  
(c). The SS model qualitatively captures the two regimes of the ATM skew as a function of maturity (d), with a cross-over around 20 days~\cite{guyon2022does,delemotte2023twoHurstmodel}, as well as the very slow decaying power-law of the ATM kurtosis \cite{potters1998financial} (e) and the behavior of the skew-stickiness ratio (f). The PDV model, on the other hand, correctly captures skewness effects (d) but fails to capture the amplitude and term structure of the ATM kurtosis (e). 
}
\label{fig:avg-smile}
\end{figure}

Fig.~\ref{fig:atm-skew} shows the absolute value of the ATM skew of the average smile for different maturities. In~\cite{guyon2022does,delemotte2023twoHurstmodel} it was shown using option market data that the ATM skew exhibits two power-law regimes pertaining to short and long maturities, with a cutoff between the two regimes around 20 days. Interestingly, we also observe this behavior on the average smile of the S\&P, which only depends on the price process, with no reference to actual option markets.  
Both the PDV model and the SS model track rather well such a behavior. The scaling of log-volatility increments characteristic of rough volatility models~\cite{gatheral2018volatility} or multifractal models \cite{bacry2013log} is in fact encoded in the model through the coefficients $\E\{|W_{j_2}|W_{j_1}x|(t)|^2\}$ included in $\Phi_4(x)$ (\ref{eq:phi4}). These consider instantaneous volatility $|W_{j_1}x|$ at scale $2^{j_1}$ and its increments at scales $2^{j_2}$. 

The SS model furthermore captures two more subtle stylized facts of option smiles: 
\begin{itemize}
    \item The ATM curvature $\kappa_T$, related to a low moment kurtosis~\cite{de2012avgsmile}, is well captured and behaves as a slow power-law of $T$ (first reported in \cite{potters1998financial}), both for the S\&P and within the SS model. The PDV model, on the other hand, strongly underestimates $\kappa_T$ (see Fig.~\ref{fig:atm-kurtosis}). 
    \item The instantaneous change of ATM volatility $\sigma_\mathrm{ATM}(T)$ induced by a change in underlying price can be linearly regressed on $\delta x$. This defines the skew-stickiness ratio $R_T$ through the following regression~\cite{bergomi2009smile,vargas2015skew}
    \[
    \delta \sigma_\mathrm{ATM}(T) = -R_T\mathrm{Skew}_T \times {\delta x} + \epsilon
    \]
    As shown in \cite{vargas2015skew}, $R_T$ has a non-trivial dependence on maturity. Fig.~\ref{fig:sksratio} shows that both the PDV model and the SS model again reproduce quite well such a dependence. 
\end{itemize}

\section{Path Shadowing Monte-Carlo \& Volatility Predictions}
\label{sec:path-shadowing}

The {\it average smile} exercise of the previous section is interesting insofar as it allows one to test the ability of various models to capture the distribution of price changes over different maturities. However, it fails to inform us on the power of the model to actually predict, at a given date, the distribution of price changes forward in time. Of course, this is what Finance is all about and we now introduce a framework to do precisely that.    

We first assume that the real world is at least approximately stationary, in the sense that it can be approximated by a statistical model with fixed, time-independent parameters. Of course, this can only be true if the model is rich enough to generate intermittent time series that superficially appear non-stationary -- such as the ones shown in Fig.~\ref{fig:stat-our-model}, with alternating periods of high and low volatility that are actually described by the {\it same} model.   

If this is the case, then given the past history $\pa{\sple{x}}$ at current time $t$, a model-free method for predicting the unknown future $\fu{\sple{x}}$ is to look for occurrences similar to $\pa{\sple{x}}$ in the historical realization of log-prices. If such occurrences can be found, what happened thereafter provides some information about the unknown future $\fu{\sple{x}}$ at the current time $t$.

Finding exact occurrences of course happens with vanishing probability. We therefore introduce an embedding $h(\pa{\sple{x}})$ that reduces the dimensionality of past trajectories and define \textit{shadowing paths} as paths $x$ whose past history $h(\pa{x})$ is close to $h(\pa{\sple{x}})$ in a certain sense. 
Furthermore, instead of scanning the historical past, we propose in this section to scan a long dataset of paths generated using the model presented in section \ref{sec:a-statistical-model}.

These shadowing paths are then used as inputs of our proposed \textit{Path Shadowing Monte-Carlo} (PS-MC) method, which allows us to obtain state-of-the-art predictions for future realized volatility. The method will be extended in the next section \ref{sec:option-pricing} to option pricing.

\subsection{The Path Shadowing Monte-Carlo Method}

First, at a given time $t$, we separate a log-price path $x\in\R^N$ into its past $\pa{x} = (x(u), u\leq t)$ and future $\fu{x} = (x(u), u>t)$
\[
x = (\pa{x}\,,\,\fu{x})
\]
Let $q(\fu{x})$ be a quantity we want to predict, for example the average realized variance in the next $T$ days $q(\fu{x}) \coloneqq \sum_{u=t+1}^{t+T} |\delta x_u|^2/T$. In the following section, we write $q(x) = q(\fu{x})$ for simplicity. 
An optimal prediction of $q(x)$ for the mean square error
as a function of the observed past $\pa{\widetilde x}$ is given by the conditional expectation
\begin{equation}
\label{eq:conditional-exp}
\E\{q(x) ~|~ \pa{x} = \pa{\widetilde x}\}
\end{equation}
with $\E$ the expectation under the true distribution $p(x)$ of log-prices. 
The goal is to estimate such conditional expectation.
 
Let us for a moment omit the conditioning on the past. The standard Monte-Carlo method estimates expectations 
using a finite number of realizations $x^1,\ldots,x^n$ drawn from $p(x)$ as
\begin{equation}
\nonumber
\frac{1}{n}\sum_{k=1}^n q(x^k)
\end{equation}
which converges to $\E\{q(x)\}$ as $n\to+\infty$ under independence of $x_1,\ldots,x_n$ and integrability of $q(x)$ .

In theory, such method could apply to estimate (\ref{eq:conditional-exp}), however it would require finding paths $x^k$ such that $\pa{x^k}=\pa{\sple{x}}$, which is all but impossible for data of finite size.

To tackle this problem, we relax strict conditioning on $\pa{\sple{x}}$ and consider paths $x$ whose past $\pa{x}$ is \textit{close} to $\pa{\sple{x}}$ in a certain sense. In order to account for possible long-range dependencies, we would like to consider a long past $\pa{\sple{x}}$. 
However, finding paths $\pa{x}$ at a given distance from $\pa{\sple{x}}$ becomes exponentially difficult as the size of the path grows -- this is the curse of dimensionality. In order to control the dimension, we consider a path embedding $h(\pa{x})\in\R^M$. Given a threshold $\eta>0$ we define the set of \textit{$\eta$-shadowing paths} as
\begin{equation}
\label{eq:shadowing-paths}
\sha = \{x\in\R^N \ | \ \|h(\pa{x}) - h(\pa{\sple{x}}) \| < \eta \}
\end{equation}
For example, $h(x) = \delta x$ considers past log-returns, and hence log-price paths up to an additive constant. Our choice of $h$ is detailed in the next section. 
The term \textit{shadowing} is freely inspired by the shadowing principle in chaotic dynamical systems~\cite{anosov1969geodesic,sinai1972gibbs,bowen1975omega,hammel1987numerical}. 
The idea is that for a certain small $\eta$, paths in $\sha$ can be considered as true realizations of the process that can be used to faithfully compute observables.

\textit{Path Shadowing Monte-Carlo}, illustrated in Fig.~\ref{fig:shadowing}, is a Monte-Carlo method on shadowing paths. It is a predictive method since shadowing paths span over the future.
Unlike a standard Monte-Carlo method, not all paths should have the same weight since $\|h(\pa{x^k}) - h(\pa{\sple{x}})\|$ is not uniform in $k$. This is to say, certain paths \textit{shadow} more accurately $\pa{x}$ than others and should be considered as more \textit{likely} to be extensions of the observed $\pa{\sple{x}}$. Path Shadowing Monte-Carlo estimates (\ref{eq:conditional-exp}) by averaging $q(x)$ on paths $x^1,\ldots, x^n$ with weights $w_1,\ldots,w_n$. This yields the following estimator
\begin{equation}
\label{eq:shadowing-MC}
\frac{1}{n}\sum_{k=1}^n w_k
q(x^k),
\end{equation}
called the Nadaraya–Watson estimator~\cite{nadaraya1964estimating,watson1964smooth}.
In the following, we choose Gaussian weights, to wit
\begin{equation}
\nonumber
w_k = c\,\exp\left[-\frac{ \|h(\pa{x^k}) - h(\pa{\sple{x}})\|^2}{2\eta^2}\right]
\end{equation}
with $c$ such that  $\frac1{n}\sum_{k=1}^n w_k = 1$. 
The set of shadowing paths $\sha$ (\ref{eq:shadowing-paths}) can be defined as the set of all paths that are one standard deviation away from $\pa{\sple{x}}$ for the Gaussian kernel. 

The following theorem proves the convergence of the estimator (\ref{eq:shadowing-MC}) under standard hypotheses.
\begin{theorem}[Path Shadowing Monte-Carlo Method]
\label{theo:shadowing-MC}
If $q$ is continuous with $\E\{|q(x)|\}<+\infty$ and the distribution $p$ of $x$ is continuous with $p(x)>0$ for all $x\in\R^N$, then given $x^1,\ldots,x^n,\ldots$ independent realizations of $x$, the Path Shadowing Monte-Carlo estimator with $h(x)=x$ converges almost surely
\[
\frac{1}{n}\sum_{k=1}^n w_k
q(x^k) \longrightarrow \E\{q(x) ~|~ \pa{x} = \pa{\sple{x}}\}
\]
for a certain limit $n\to+\infty$ and $\eta \to 0$. 
\end{theorem}
The proof is in Appendix \ref{app:proofs}. 
The continuity assumptions as well as the assumption that $p>0$ are technical and can be softened; note that a $p_\theta$ of the form (\ref{eq:exp-model}) satisfies them. We refer the reader to~\cite{hansen2008uniform} for convergence theorems under more generic hypotheses.

Path Shadowing Monte-Carlo can be seen as a kernel method on log-price paths using a Gaussian kernel. It is a fully non-local method in the sense that the collected paths $x^k$ may be far away in the past from $\pa{\sple{x}}$, possibly in a generated dataset of paths, contrary to non-local means~\cite{buades2011non} that only consider neighborhoods of a patch. Such non-locality ensures in practice the independence condition in theorem \ref{theo:shadowing-MC}.

\subsection{Generating Shadowing Paths}
\label{subsec:choice-of-parameters}

Collecting enough shadowing paths from the historical realization of S\&P is unrealistic. 
Indeed, financial price paths are noisy, two different observed paths are likely to be distant from one another, signifying a high-entropy underlying process. Thus, on limited data,
the set $\sha$ will contain almost no paths for reasonable values of $\eta$, required to be small for the method to converge.

We thus propose to scan for paths $x^1,\ldots x^n$ in a long generated dataset of log-prices, allowing us to take $n\gg1$ and selective shadowing threshold $\eta\ll1$. This however immediately introduces a modelling error, due to the fact that we are estimating (\ref{eq:conditional-exp}) where $\E$ is now the expectation with regard to the model distribution and not the true distribution $p(x)$.

A good model should generate trajectories that capture the long-range dependencies in order for the shadowing paths to have predictive power on the future of $\pa{\sple{x}}$. It should also generate a variety of realistic scenarios in order to find enough paths in $\sha$ for small $\eta$. 
As discussed in section \ref{sec:a-statistical-model} the SS model precisely meets these requirements: 
it promotes diversity of generated paths through entropy maximization while reproducing main statistical properties of financial prices.
Furthermore, should our generative algorithm produce occasionally unrealistic paths, 
they
would be given a very small weight $w$ and would not contribute to the estimation of $q(x)$. Shadowing Monte-Carlo is thus robust to outliers in the generated set of paths. 
Note also that this dataset can be generated once and can be scanned again and again for several prediction dates.

A crucial point for PS-MC to give good results is to understand how the path embedding $h$ affects the notion of path proximity. Such embedding should be chosen adequately. It should pick relevant features of $\pa{x}$ to predict the quantity of interest $q(x)$, while remaining low-dimensional such that enough paths can be found in $\sha$ for small $\eta$. The naive embedding $h(x) = \big(\delta x (t), -M+1 \leq t \leq  0\big)\in\R^M$ limits the past horizon $M$, since $M$ is then also the dimensionality of $h$. 

We propose an embedding $h$ that again leverages the scale-invariance of $x$ in the same way our Scattering Spectra framework does, and incorporates the influence of distant past while remaining low-dimensional. 
It consists of multi-scale increments in the past with a power-law decaying weight
\begin{equation}
\label{eq:choice-of-h}
h_{\alpha,\beta}(x) = \Big( \frac{ x(t) - x(t-\ell)}{\ell^\beta} 
~ , ~ \ell=\lfloor \alpha^m \rfloor  ~ , ~ m=1,2,\ldots \Big)
\end{equation}
for a certain $\alpha>1$ so that the past is progressively coarse-grained, and $\beta\geq0$ so that the far away past is progressively discounted. 
For a given $\pa{\sple{x}}$ we choose
$\eta$ to be equal to $\widehat \eta \|h(\pa{\sple{x}})\|$, which amounts to normalize the distance $\|h(\pa{x}) - h(\pa{\sple{x}})\|$ by $\|h(\pa{\sple{x}})\|$. 
Such $h$ is a discretization of a continuous embedding that satisfies scaling and dilation equivariance, see Appendix \ref{app:choice-of-h}. In practice we truncate the progression $\lfloor \alpha^1 \rfloor, \lfloor \alpha^2 \rfloor,\ldots$ in order for the span of $h$ to be bounded by $126$ trading days in the past (corresponding to half a year). 

The main parameters are thus $\alpha,\beta$ and $\widehat \eta$. 
Parameter $\alpha$ determines the dimensionality of the path embedding $h_{\alpha,\beta}$, small $\alpha$ yields high-dimensional embedding. Parameter $\beta$ (not to be confused with the parameters $\beta_{0,1,2}$ of the PDV model) rules the relative importance of distant past to recent past in the selection of shadowing paths. Large $\beta>0$ means that the recent past bears more weight.
Finally, there is a bias-variance trade-off in the choice of $\widehat \eta$. When $\widehat \eta\ll1$ only the closest path will be used for averaging, leading to large variance estimator (\ref{eq:shadowing-MC}). When $\widehat \eta\gg1$ then all paths are averaged uniformly, including paths whose past $\pa{x}$ has nothing to do with $\pa{\sple{x}}$, thus deteriorating the bias of our PS-MC estimator.

\subsection{Volatility Prediction}
\label{subsec:volatility-prediction}

As a meaningful application of Path Shadowing Monte-Carlo, we consider in this section the prediction of the future daily realized variance over $T$ days, for several values of $T$:
\begin{equation}
\nonumber
q_T(x) = \frac{252}{T}\sum_{u=t+1}^{t+T} |\delta x_u|^2.
\end{equation}
We consider all 2112 dates $t$ from January 2015 to March 2023. For each of them we consider the realized variance over $T=1,7,25,75,150$ days.

Our PS-MC method (\ref{eq:shadowing-MC}) uses paths generated from the model presented in section \ref{sec:a-statistical-model}. We compute the Scattering Spectra statistics (\ref{eq:scattering-spectra}) on observed 3772 S\&P log-prices from January 2000 to December 2014, such that all our predictions are {\it out-of-sample}. From such statistics we generate 32\,768 trajectories of same size 3772 (see Fig.~\ref{fig:generated-samples} for examples), that represents $n\approx$ 115 million paths $x^1,\ldots,x^n$ of size 126+150 days. 
For a given $\pa{\sple{x}}$ we scan such dataset and select the 50\,000 closest paths in the sense of the distance induced by embedding (\ref{eq:choice-of-h}), parameterized by $\alpha,\beta$. 
While this scanning step is computationally demanding, it can be fully parallelized.
We then perform a weighted average on those closest paths, parameterized by $\widehat \eta$.  

Parameters $\alpha,\beta,\widehat \eta$ are calibrated using our generative model itself, in order to avoid any over-fitting on the limited train data from S\&P. 
We choose those parameters such that $q_T(x)$ is optimally predicted within the model. More precisely,
we take 1100 snippets $\pa{\sple{x}}$ from the generated dataset, for which we have access to $\fu{\sple{x}}$. We obtain an estimate of $q_T$ for these 1100 dates and $T=7,25,75,150$. We then choose the best $\alpha,\beta,\widehat \eta$ so as to maximize the $R^2$ score between prediction and true values of the SS model itself. This yields the following optimal parameters: $\alpha=1.15,\beta=0.9,\widehat \eta=0.075$ and a path embedding $h_{\alpha,\beta}$ of dimension 34. 

Let us note that these optimal values barely change when predicting realized variance at different maturities. This is because of the scale-invariance of both the model and of the path embedding. Note also that $\alpha=1.15$ in (\ref{eq:choice-of-h}) means that the values $\ell=1,2,3,4$ appear multiple times. This amounts to ascribing an even larger weight to small time lags. 

\begin{table}[!h]
\centering
\begin{tabular}{|c|c|c|c|c|c|}
\hline
Number of days $T$ & $1$ & $7$ & $25$ & $75$ & $150$
\\
\hline
\hline
Benchmark & -0.16 & 0.43 & -0.05 & -0.58 & -0.79
\\
\hline 
PDV (SS regression) & 0.35 & 0.53 & 0.17 & -0.44 & -0.95
\\
\hline 
PDV (linear regression) & \textbf{0.36} & 0.55 & 0.29 & 0.0 & -0.07
\\
\hline
SS Path Shadowing & 0.32 & \textbf{0.56} & \textbf{0.33} & \textbf{0.07} & \textbf{0.01}
\\
\hline
\end{tabular}
\caption{Prediction of realized daily volatility ($R^2$ score).
The PS-MC method based on the Scattering Spectra (SS) model outperforms the recently introduced PDV model
at all time-scales $\geq$ 7 days, and upholds predictive power up to a period of $\approx$ 150 days. For $T=1$ day, however, the PDV model performs best. The benchmark is the realized variance on the $T$ previous days. The PDV model was calibrated using two different methods, see Appendix \ref{app:pdv}.
}
\label{tab:vol-prediction}
\end{table}

The prediction quality is measured through the $R^2$ score on future volatility estimates and are shown in Table \ref{tab:vol-prediction} for different maturities $T$. As a simple benchmark we consider  the realized variance on the $T$ previous days as a predictor of $q_T(x)$. As a second, more challenging, benchmark we consider the recent Path-Dependent Volatility (PDV) model introduced in~\cite{guyon2021volatility}, which reads
\begin{equation}
\nonumber
\sqrt{q_T(x)} = \beta_0 + \beta_1 R_{1,t} + \beta_2 \sqrt{R_{2,t}}
\end{equation}
\begin{equation*}
\text{with} ~~ R_{1,t} = k_1\star \delta x(t) ~ , ~ R_{2,t} = k_2\star |\delta x|^2(t),
\end{equation*}
where $k_1$ and $k_2$ are both linear combinations of exponential kernels, acting on past returns and past absolute returns\footnote{Note that the authors of Ref.~\cite{guyon2021volatility} estimate realized daily variance using 5-min ticks for better estimation, but the scores we obtain with daily ticks are actually similar for the longest maturities $T=3$ and $T=5$ that were tested in their study. Hence, we do not think that using 5-min ticks would substantially change the conclusions reached below for $T \geq 7$ days.}. 
We take the very same kernels as specified in~\cite{guyon2021volatility} but we determine the linear regression coefficients $\beta_0,\beta_1,\beta_2$ in two different ways, see Appendix \ref{app:pdv}. The first one amounts to reproducing as well as possible the Scattering Spectra, as we did in section \ref{sec:avg-smile} to reproduce the average smile. The corresponding values of $\beta_{0,1,2}$ are given in Table~\ref{tab:pdv-unconditional-calib}. Unfortunately, this leads to rather poor results as far as volatility prediction is concerned, specially for large $T$, see Table \ref{tab:vol-prediction}. 

In order to favour the PDV model as much as possible, we used a different calibration aiming at best regressing $\sqrt{q_T}$ for each maturity $T$ separately, on the same train period as for the SS model, i.e. from January 2000 to December 2014, see Table~\ref{tab:pdv-conditional-calib} for calibrated values.

Using the very same shadowing paths for all maturities, our Path Shadowing Monte-Carlo method outperforms both the naive benchmark and the PDV model for all maturities from $T > 7$ days, and ties with PDV for $T=7$ days, see Table \ref{tab:vol-prediction}. In particular, our method upholds predictive power up to $\approx$ 150 days, which 
none of the two other methods are capable of. This is, we believe, due to the scale-invariance of both the SS model and our choice of path embedding (\ref{eq:choice-of-h}). Again, we insist on the fact that the PDV model parameters are refitted for each maturity $T$ whereas the SS model is calibrated once and for all. 
Of course, the SS model contains many more parameters than the PDV model, and we have not attempted to modify the shape of the kernel $k_1$ and $k_2$ introduced in~\cite{guyon2021volatility}, so there is clearly room for improvements there as well, which could be investigated in future work.

These results vindicate both the generative model presented in section \ref{sec:a-statistical-model} and the PS-MC method of the present section. In particular the Scattering Specra generative model, now based on 182 parameters determined on the period 2000-2014, is {\it not} over-fitting the training dataset.

\section{Option Pricing \& Trading Games}
\label{sec:option-pricing}

In section \ref{sec:avg-smile}, we have used the Hedged Monte-Carlo method to price options {\it unconditionally} within the SS model, i.e. by averaging over all possible price paths of a given length $T$. This allowed us to obtain {\it average smiles} as a function of maturity, which we compared to those obtained using the same procedure but with real S\&P trajectories. 

Now, at a given date, option prices reflect anticipations of the market, conditioned on present market conditions -- in particular the past price trajectory -- and any available information about the future, such as earning announcements, dividends, political events, etc. Of course, such events cannot be directly captured by a purely statistical model, however faithful it might be. Still, it is interesting to run the exercise of pricing option smiles that anticipates the future solely based on the past of underlying price process.

In this section we investigate this question by combining the SS model presented in section \ref{sec:a-statistical-model} with Path Shadowing introduced in section \ref{sec:path-shadowing}.
Option prices are then obtained by upgrading the Hedged Monte-Carlo method~\cite{potters2001HMC} with, as an input, shadowing paths generated by the model.
The overall level of the resulting smiles at time $t$ is nothing but the prediction of the future realized volatility for $t' \in [t,t+T]$, which was already investigated in the previous section. We now extend such prediction to the full shape of the smile. 
We assess the quality of our smiles by trading a buy-sell signal on options whenever the model option smile is telling us that the option is under-priced or over-priced compared to another smile model, or of the option market itself.

\subsection{Path Shadowing Hedged Monte-Carlo}

Hedged Monte-Carlo (HMC)~\cite{potters2001HMC} enables one to use time series of prices to compute the option prices. It iteratively determines the optimal price 
and hedging strategy by minimizing the expected financial risk of a portfolio containing the option to be priced and its hedge, at all times $t=T-1,T-2,\ldots,0$. The expected risk is computed as an average over paths, which in the present context are the
shadowing paths, based on the notion of distance induced by the path embedding $h$ (\ref{eq:choice-of-h}). 
This defines the Path Shadowing Hedged Monte-Carlo (PS-HMC) that can be used in a versatile way to price any derivative contract. We use the same Gaussian weights (\ref{eq:shadowing-MC}) and the very same parameters $\alpha,\beta,\eta$ detailed in section \ref{subsec:volatility-prediction}, that are optimal for volatility prediction within the model itself.

Fig.~\ref{fig:cond-smile} shows the resulting smiles as a function 
of rescaled log-moneyness,
for different maturities and at 3 typical dates. As one would have hoped, the level, but also the slope and the curvature of those smiles do depend on the chosen date, and more precisely on the actual path trajectory of the price before that date.  

\begin{figure}[!h]
\centering
\includegraphics[width=0.95\linewidth]{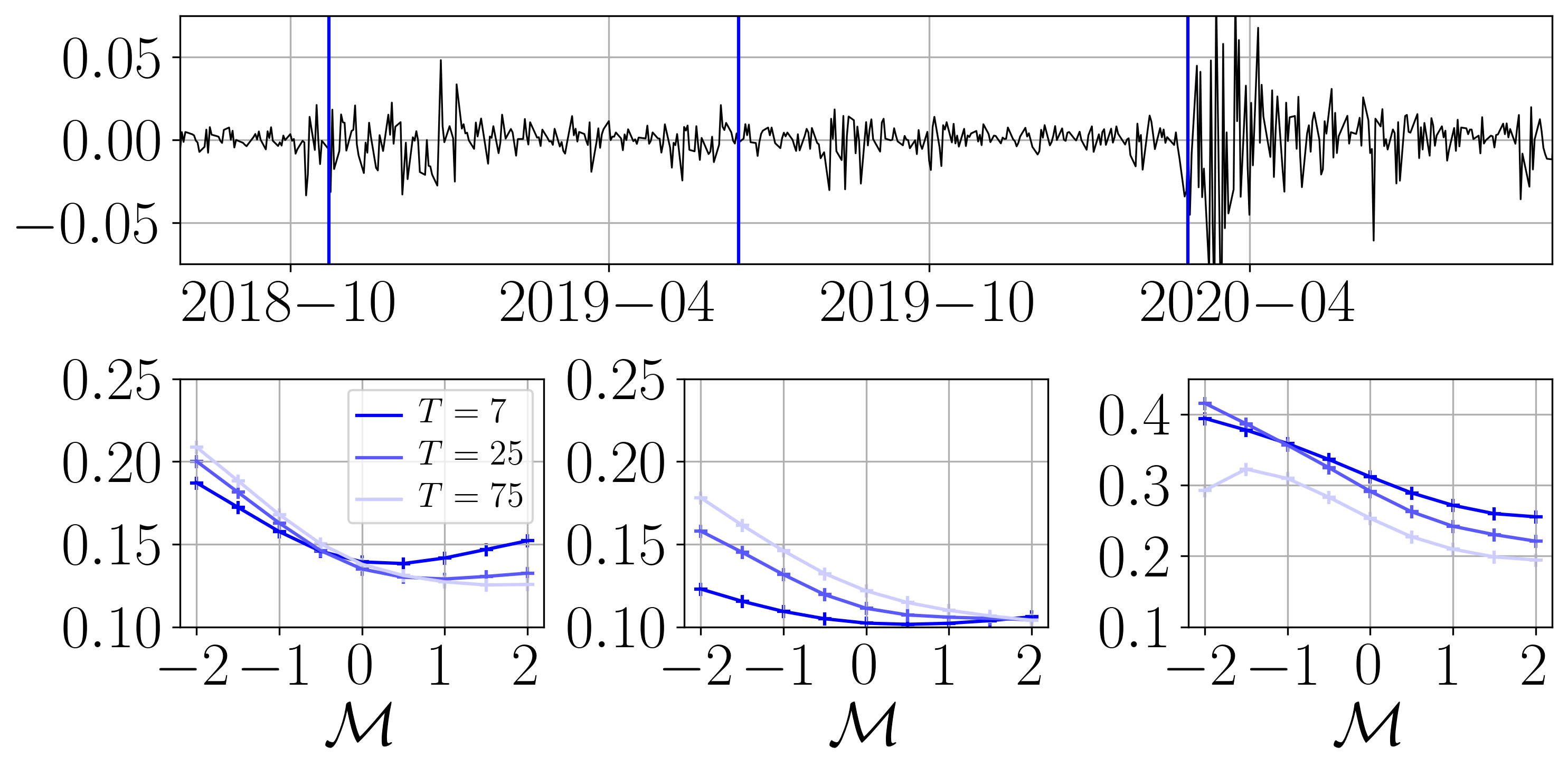}
\caption{
Conditional smiles in a Scattering Spectra (SS) model at 3 different dates 2018-10-23, 2019-06-14 and 2020-02-26.
They are computed through hedged Monte-Carlo on shadowing paths generated using the SS model.
Note that the level, the slope and the curvature of those smiles all strongly depend on the chosen date. 
}
\label{fig:cond-smile}
\end{figure}

\subsection{A Trading Game}

In order to assess the quality of the smiles predicted by the SS model, we set up the following trading game. 
We trade call options at several dates $t$ on the option market. We neglect bid-ask spread and consider the option price $\mathcal{C}^\text{\sc{MKT}}(t,T,K)$ to be the quoted mid-price. We denote $\sigma^\text{\sc{MKT}}(t,T,K)$ the observed implied volatility computed from $\mathcal{C}^\text{\sc{MKT}}(t,T,K)$ and $\sigma^\text{model}(t,T,K)$ the implied volatility computed within the model that we decide to trade with. 

We test the following trading strategy: buy the corresponding option 
from the market
whenever we deemed it under-priced, i.e. $\sigma^\text{\sc{MKT}}(t,T,K)<\sigma^\text{model}(t,T,K)$ or sell it if we deem it over-priced, i.e. $\sigma^\text{\sc{MKT}}(t,T,K)>\sigma^\text{model}(t,T,K)$. We then follow the hedged option until maturity and register the corresponding profit or loss associated to the trade. 

The buy-sell signal of such strategy is thus given by 
\[
\eps_t = \left\{
\begin{array}{ll}
    +1 & \mbox{if } 
    \sigma^\text{\sc{MKT}}(t,T,K) < \sigma^\text{model}(t,T,K),
    \\
    -1 & \mbox{if }
    \sigma^\text{\sc{MKT}}(t,T,K) > \sigma^\text{model}(t,T,K).
\end{array}
\right.
\]
The un-hedged forward $\mathrm{P\&L}_t(T,K)$ of one transaction is obtained as
\begin{equation}
\nonumber
\mathrm{P\&L}_t(T,K) = v_t\eps_t
\bigg((e^{x_{t+T}}-K)_+ - \mathcal{C}^\text{\sc{MKT}}(t,T,K)\bigg)
\end{equation}
where $v_t$ is the volume of option traded. 
In order to remove non-stationary effect due to the long-term change in the value of the underlying, we take $v_t = 100/S_t=100e^{-x_t}$ which amounts to trade options on percentage of variation of the underlying rather than the underlying itself. 

To reduce the variance of the strategy, we hedge the option using a simple Black-Scholes delta-hedge with a constant volatility $0.2$. Such delta-hedge gives zero profit on average but reduces greatly (though not optimally, see~\cite{potters2001HMC,bouchaud2003theory} for details) the variance of $\mathrm{P\&L}_t$. 

In the following, we will consider the model smile $\sigma^\text{model}$ to be given either by the smile computed in the SS model using PS-HMC (model=SS) or the smile computed in a Path-Dependent Volatility model \cite{guyon2021volatility} (model=PDV), both using HMC. The two models (PDV and SS) are calibrated using the same data in the train period i.e. January 2000-December 2014. 
As in the previous section, the PDV model parameters are furthermore {\it re-optimized} for each maturity $T$ independently as in Table \ref{tab:pdv-conditional-calib}, whereas the SS model is parameterized once and for all with the Scattering Spectra determined in the train period.

The trading game is then in both cases played out-of-sample, for all 2112 dates $t$ from January 2015 to May 2023. 
We choose $5$ maturities $T=\, 8,\, 25,\,50,\,75,\,150$ and $9$ strikes at constant rescaled log-moneyness $\mness=-2,-1.5,-1,-0.5,0,0.5,1,1.5,2$.  

Detailed $\mathrm{P\&Ls}$ over 3 different periods of 3 year each are shown in Fig.~\ref{fig:smile-prediction-avg}. Their variance across dates $t$ is shown in Fig.~\ref{fig:smile-prediction-std} in Appendix \ref{app:trading-game}. For most maturities and strikes, the trading game using the SS model yields positive $\mathrm{P\&Ls}$ and clearly outperforms the trading game using the PDV model. In fact, one can directly play the SS model against the PDV model without any reference to the actual option market, fully confirming that the SS model outperforms PDV model  for almost all maturities and strikes, see Appendix \ref{app:trading-game} and in particular Fig.~\ref{fig:pdv_ssb_game}. 

Since the $\mathrm{P\&Ls}$ are of the same order of magnitude across different strikes and maturities, we average them over all the maturities and strikes. Table \ref{tab:smile-prediction} shows such grand averages and reveals that the trading game using the SS model is significantly more profitable 
than using the PDV model, with furthermore less variance across different periods. This is confirmed by the aggregated $\mathrm{P\&Ls}$ over time, see Fig.~\ref{fig:smile-prediction-aggregated}.

\begin{table}[!h]
\centering
\begin{tabular}{|c||c|c|c|c|}
\hline
& full period & $2015$-$2017$ & $2018$-$2020$ & $2021$-$2023$ 
\\
\hline
\hline
PDV & 
0.071 $\pm$ 0.03 & 
0.13 $\pm$ 0.01 & 
0.0 $\pm$ 0.04 &  
0.10 $\pm$ 0.02
\\
\hline 
SS & 
\textbf{0.13} $\pm$ 0.03 & 
\textbf{0.14} $\pm$ 0.02 & 
\textbf{0.14} $\pm$ 0.04 & 
\textbf{0.11} $\pm$ 0.02
\\
\hline
\end{tabular}
\caption{Average $\mathrm{P\&L}$ of the trading game against the S\&P market using the PDV model (re-optimized for each $T$) or using the Scattering Spectra (SS) model, both calibrated using data from 2000 to 2014. Detailed $\mathrm{P\&Ls}$ are shown in Fig.~\ref{fig:smile-prediction-avg}.}
\label{tab:smile-prediction}
\end{table}
\begin{figure}[!h]
\centering
\includegraphics[width=0.9\linewidth]{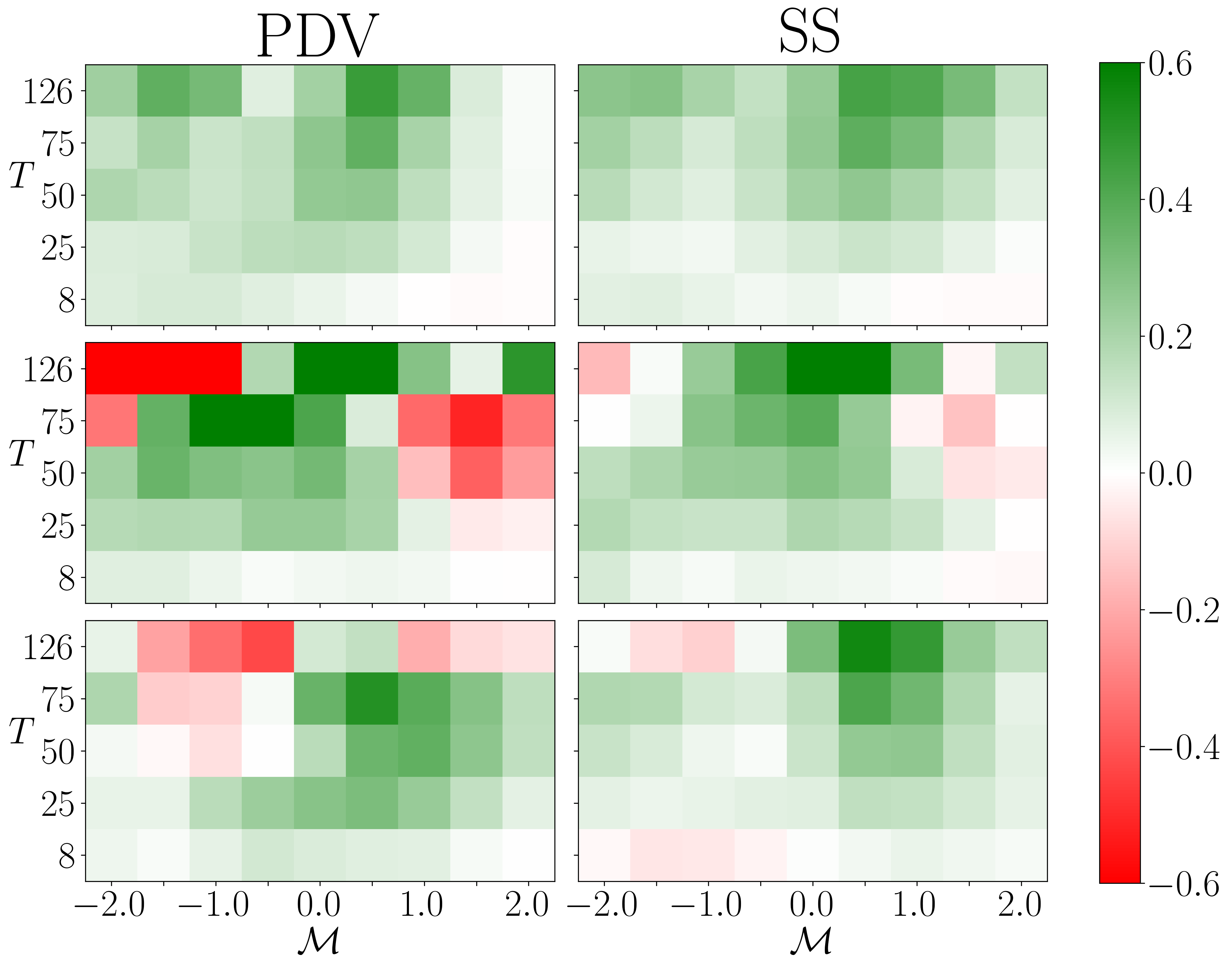}
\caption{
$\mathrm{P\&Ls}$ of the trading game against the S\&P market with a PDV model (re-optimized for each $T$) or with the Scattering Spectra (SS) model, both calibrated using data from 2000 to 2014. Each heatmap corresponds to a 3 years period, from top to bottom (2015-2017, 2018-2020, 2021-2023). }
\label{fig:smile-prediction-avg}
\end{figure}
\begin{figure}[!h]
\centering
\includegraphics[width=1\linewidth]{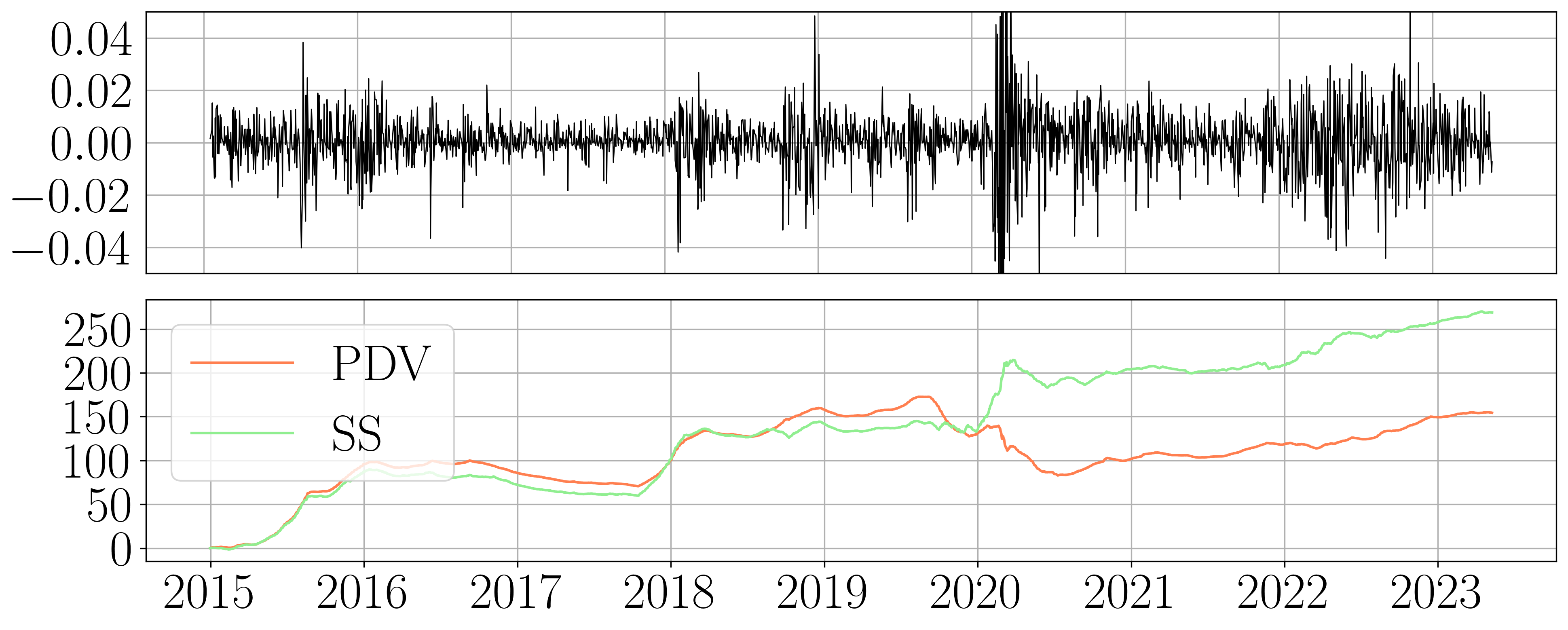}
\caption{(Bottom) Aggregated $\mathrm{P\&Ls}$ of the trading game against the S\&P market, played from 2015 to 2023, with a PDV model (re-optimized for each $T$) or with the Scattering Spectra (SS) model, both calibrated using data from 2000 to 2014.
(Top) Time series of S\&P log-return on the trading game period.
}
\label{fig:smile-prediction-aggregated}
\end{figure}

\newpage

\section{Conclusion}

We presented a statistical model of financial prices based on the Scattering Spectra introduced in our previous work~\cite{morel2022scale}. Scattering Spectra are multi-scale analogues of the standard skewness and kurtosis. Such a model achieves a tradeoff between accuracy and diversity. It captures main statistical properties of prices as well as several scaling properties of option smiles. As a maximum entropy model with a small number of constraints, our Scattering Spectra model also avoids over-fitting.

We then introduced Path Shadowing Monte-Carlo (PS-MC) which enables building models of forward looking probabilities $p(x|\pa{x})$ from any generative model of $p(x)$. Combined with our statistical model of prices, PS-MC provides state-of-the-art volatility predictions. 
Shadowing paths can also be used to obtain option smiles that depend only on the distribution of the price process. 
A trading game allowed us to show that the Scattering Spectra model correctly anticipates future price movements and favourably competes with the (simplest version of) the recently introduced Path-Dependent Volatility model. Extending our analysis to other underlying assets, such as single stocks, is of course important to validate our conclusions when price moves are more idiosyncratic.

Beyond prediction, Path Shadowing may be a way of tackling other burning questions in Finance, such as \textit{typicality}: how typical or atypical is a given sequence of price changes?
As far as Path Shadowing Monte-Carlo is concerned, one limitation of the method is that it requires to scan a large dataset of generated paths for delivering good performances. This scanning step could be done more efficiently. In particular, could one find `typical' price paths that should be frequently selected for prediction in order to save intensive scanning efforts? 
Another highly relevant extension is towards the description of multivariate time series. We hope to address these issues in the near future using the methods introduced in this work. 

\paragraph*{Acknowledgments} We want to thank C. Aubrun, N. Cuvelle-Magar, S. Crépey, J. Gatheral, F. Guth, J. Guyon, B. Horvath, G. Pagès, M. Potters \& V. Vargas for many insightful conversations on these topics.

\bibliographystyle{IEEEtran}
\bibliography{IEEEabrv,biblio.bib}

\appendices

\section{Sampling Algorithm}
\label{app:sampling}

Drawing samples from the Scattering Spectra model (\ref{eq:exp-model}) could be performed through an exact algorithm~\cite{rolf1998microsampling}. This is however computationally expensive when the number of statistics $\Phi$ is not small. Instead we use an approximate algorithm based on gradient descent~\cite{bruna2019multiscale}. We consider the set of trajectories $x$ having approximately the correct statistics
\begin{equation}
\nonumber
\Omega_\eps = \{ x\in\R^N \ | \ \| \Phi(x) - \Phi(\widetilde x)\|  \leq \eps  \}.
\end{equation}
The width $\epsilon$ accounts for the statistical error between $\Phi(\sple{x})$ and the true $\E\{\Phi(x)\}$. We chose $\epsilon=10^{-3}\|\Phi(\sple{x})\|$.
An approximate sampling of model (\ref{eq:exp-model}) is performed by sampling trajectories in $\Omega_\eps$~\cite{bruna2019multiscale}.
The algorithm is initialized with a realization $x\in\R^N$ of a Gaussian noise which has a maximum entropy distribution. Then we minimize the loss $\ell(x) = \|\Phi(x) - \Phi(\sple{x})\|^2$ through gradient descent
\[
x \longleftarrow x - \gamma \nabla \ell(x).
\]
until the condition $\ell \leq \eps$ is reached. 
The gradient descent is implemented with the L-BFGS-B algorithm.

\section{Wavelets}
\label{app:wavelet}

We impose that the wavelet $\psi$ satisfies the following energy conservation law called the Littlewood-Paley equality
\begin{equation}
\label{littlewood}
\forall \om > 0~~,~~ \sum_{j=-\infty}^{+\infty} |\widehat {\psi} (2^j \om)|^2 = 1 .
\end{equation}
A Battle-Lemarié wavelet \cite{battle1987block, lemarie1988ondelettes} is an example of such wavelet, shown in Fig.~\ref{fig:wavelet}. 
It has an exponential decay away from $t=0$ and has $m = 4$ vanishing moments
\begin{figure}[!h]
\centering
\includegraphics[width=0.9\linewidth]{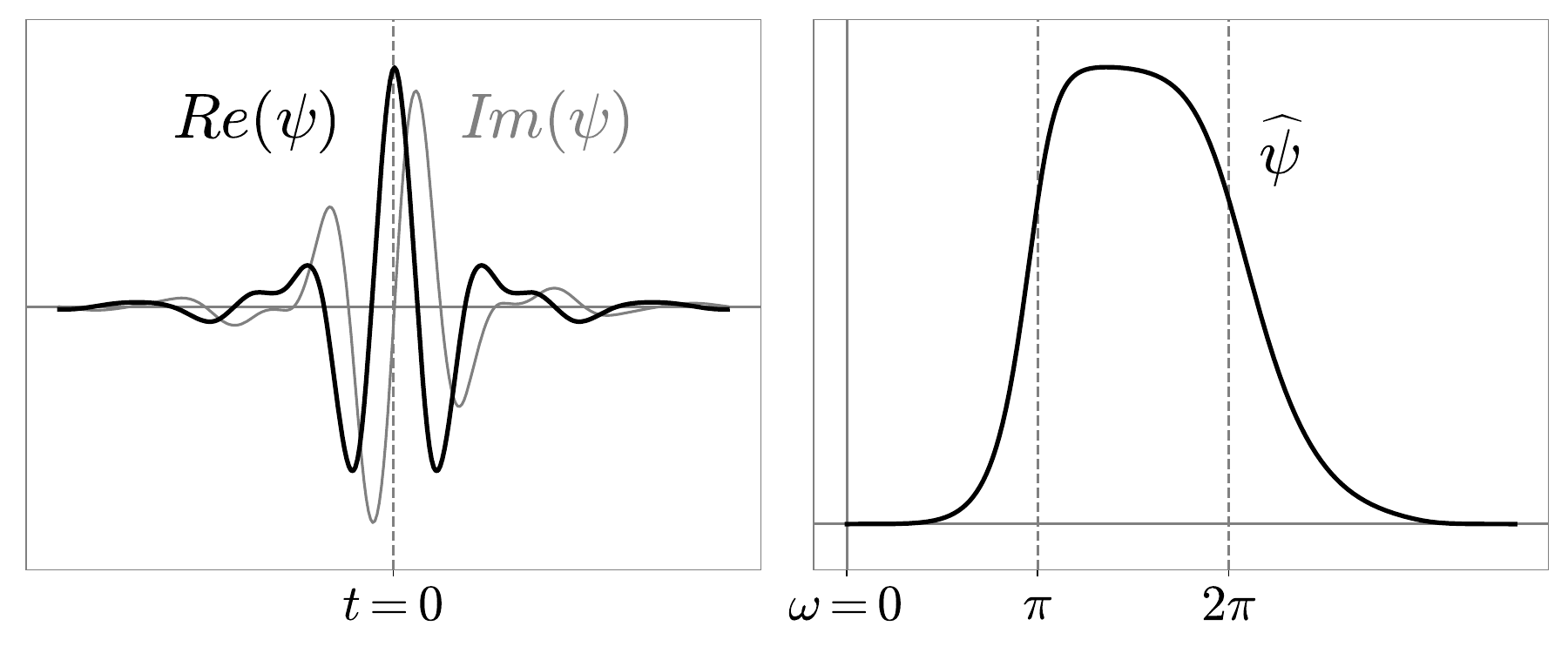}
\caption{Left: complex Battle-Lemarié wavelet $\psi(t)$ as a function of $t$.
Right: Fourier transform $\widehat{\psi}(\om)$ as a function of $\om$.}
\label{fig:wavelet}
\end{figure}

The wavelet transform is computed up to a largest scale $2^J$ which is smaller than the signal size $N$.
The signal lower frequencies in $[-2^{-J} \pi, 2^{-J} \pi]$ are captured by a low-pass filter $\varphi_J(t)$ whose Fourier transform is
\begin{equation}
\label{lowpass}
\widehat{\varphi}_J (\om) = {\Big(\sum_{j=J+1}^{+\infty} |\widehat{\psi} (2^j \om)|^2\Big)^{1/2}} . 
\end{equation}
One can verify that it has a unit integral $\int \varphi_J(t)\, {\rm d}t = 1$. 
To simplify notations, we write this low-pass filter as a last scale wavelet
$\psi_{J+1} = \varphi_J$, and $W_{J+1}x(t) = x \star \psi_{J+1}(t)$. By applying the Parseval formula,
we derive from (\ref{littlewood}) that for all $x$ with $\|x\|^2 = \int |x(t)|^2\,{\rm d}t < \infty$
\begin{equation}
\nonumber
\|W x \|^2 = \sum_{j=- \infty}^{J+1}\|x \star \psi_j \|^2 = \|x\|^2 .
\end{equation}
The wavelet transform $W$ preserves the norm and is therefore invertible, with a stable inverse.

Properties of signal increments are carried over to wavelet coefficients by observing that wavelet coefficients are obtained by filtering signal increments
$\delta_j x(t) = x(t) - x(t - 2^j)$ with a dilated integrable filter:
\begin{equation}
\label{dXvsWx}
x \star \psi_j (t) = \delta_j x \star \theta_j (t)~~
\mbox{where}~~ \theta_j (t) = 2^{-j} \theta(2^{-j} t),
\end{equation}
where filter $\theta$ is obtained from $\psi$ through
$\widehat{\theta} (\om) = \widehat{\psi}(\om)\, / \, (1 - e^{-i \om})$. 
This is because $1-e^{i2^j\om}$ is the Fourier transform of $\delta(t) - \delta(t-2^j)$, the filter that creates increments.
From (\ref{dXvsWx}) we get that if $\delta_j x(t)$ is stationary then $x \star \psi_j (t)$ is also stationary.

\begin{figure}[!h]
\centering
\begin{subfigure}[b]{0.49\textwidth}
    \centering
    \includegraphics[width=\textwidth]{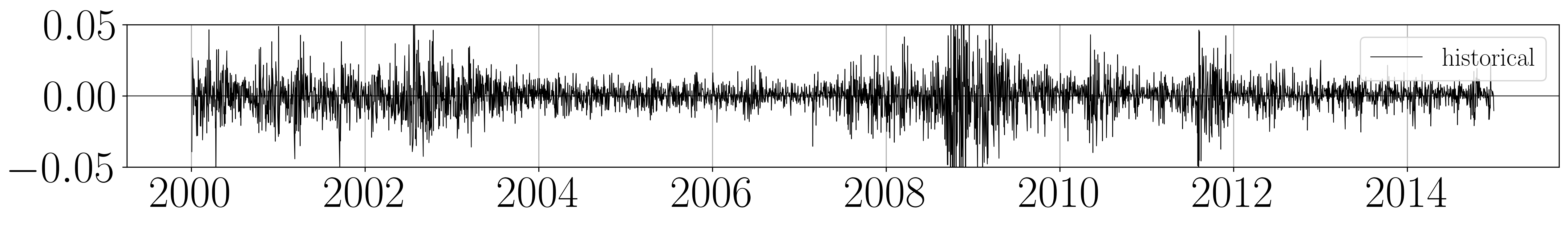}
\end{subfigure}
\vskip\baselineskip
\begin{subfigure}[b]{0.49\textwidth}   
    \centering 
    \includegraphics[width=\textwidth]{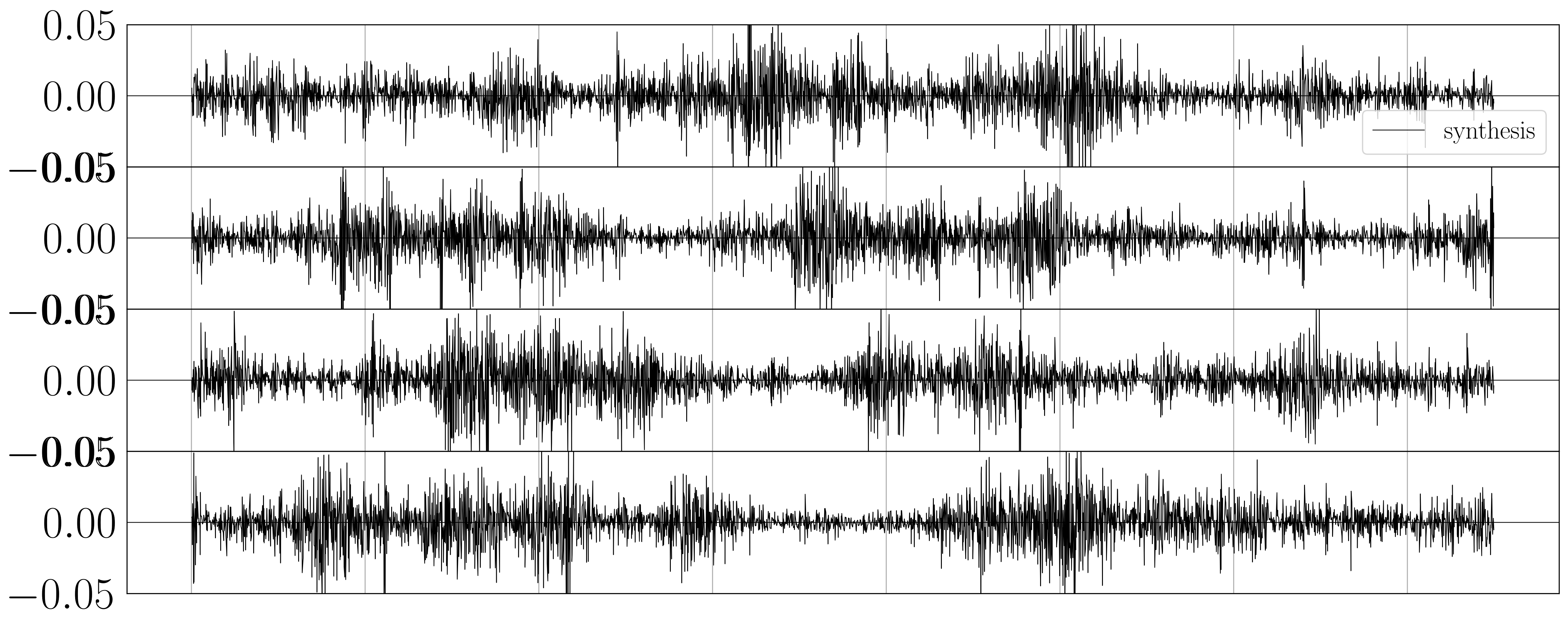}
\end{subfigure}
\caption{(Top) Historical S\&P log-returns from 2000 to 2014. (Bottom) Examples of generated syntheses from the Scattering Spectra model, which are scanned for shadowing paths.}
\label{fig:generated-samples}
\end{figure}

\section{Smile Sensitivity in the Scattering Spectra Model}
\label{app:smile-sensitivity}

The Scattering Spectra model defined in (\ref{eq:exp-model}) depends only on the estimated values $\Phi(\sple{x})$ of the Scattering Spectra estimated on a single realization $\sple{x}$ of S\&P. This models gives an unconditional smile shown in Fig.~\ref{fig:smile-avg-ss}.

We are interested in the change in this unconditional smile in the case where the statistics $\Phi(\sple{x})$ change significantly. 
Thanks to the interpretation of Scattering Spectra coefficients, we can see what happens to the smile if the market is `more skewed' or `more kurtic' for example.

Here we focus on the shape of the smiles. Of course, changing the amplitude of $\Phi_2$, which is equivalent to increasing the overall volatility of the model, only moves the overall level of the smile up or down.

These sensitivities in general can be seen as amplifying or reducing the departure of the price process from a Gaussian process $x_\text{Gaussian}$, also called Black-Scholes model, with the same average volatility, meaning the same value of $\Phi_2(\sple{x})$. 
\[
\text{new~statistics} = (1-\lambda) \Phi(x_\text{Gaussian}) + \lambda \Phi(\sple{x}).
\]
If $\lambda<1$ the corresponding model should be `closer' to a Black-Scholes model, if  $\lambda>1$ the corresponding model should become less `Gaussian'.
Fig.~\ref{fig:smile-sensitivity} shows 4 directions of change that are detailed below.
\\
\noindent\textbf{Skewness $|\Phi_3(x)|$.}
The skewness coefficients $\Phi_3(x_\text{Gaussian})$ should be zero for a Gaussian process. 
We consider a model of $x$ with modified statistics
\[
\Phi_3(x) = \lambda \Phi_3(\sple{x})
\]
for 3 values of $\lambda$.
For $\lambda=0$, the modeled process is not skewed, meaning that an increment trajectory $\delta x$ is equally likely as a trajectory $-\delta x$. Unsurprisingly, we get smiles that are symmetrical at $\mness=0$. 
For $\lambda=1$ we get the same smile as in the Scattering Spectra model of S\&P. For $\lambda=1.3$, the smile has a higher downward slope, as expected.
\\
\noindent\textbf{Time-asymmetry Im$\,\Phi(x)$.} Our representation $\Phi$ is complex-valued. While $\Phi_1,\Phi_2$ are real, skewness $\Phi_3$ and kurtosis $\Phi_4$ may have non-zero imaginary parts that were shown to characterize certain types of time-asymmetry.
We consider a model $x$ whose statistics are 
\[
\Phi(x) = \text{Re}\, \Phi(\sple{x}).
\]
It is thus time-reversible, i.e. a trajectory $x(t)$ is equally likely as $x(-t)$. 
We notice that its smiles are symmetrical, but compared to the previous case, these are not symmetrical around $\mness=0$ but around values $\mness>0$ depending on the maturity. 
This is consistent since the process still has non-zero skewness $|\Phi_3(x)|$. 
\\
\noindent\textbf{Kurtosis $\Phi_1$.} For a Gaussian process, $\Phi_1(x_\text{Gaussian})=\pi/4$. We consider a model $x$ with modified statistics
\[
\Phi_1(x) = (1-\lambda) \frac{\pi}{4} + \lambda\Phi_1(\sple{x}).
\]
For $\lambda=0.5$, the model is less kurtic than the S\&P, it shows smiles that tend to flatten around a straight line with negative slope. 
For $\lambda=1.75$, the model is more kurtic and the smiles have more curvature, as expected.
\\
\noindent\textbf{Kurtosis $\Phi_4$.} 
We consider a model $x$ with modified statistics
\[
\Phi_4(x) = (1-\lambda) \Phi_4(x_\text{Gaussian}) + \lambda\Phi_4(\sple{x}).
\]
The change in smiles for two different values $\lambda=0.5,1.5$ seem small compared to the other effects presented, however such changes impacts a lot the trajectories. 

\begin{figure}[h]
\centering
\includegraphics[width=1.0\linewidth]{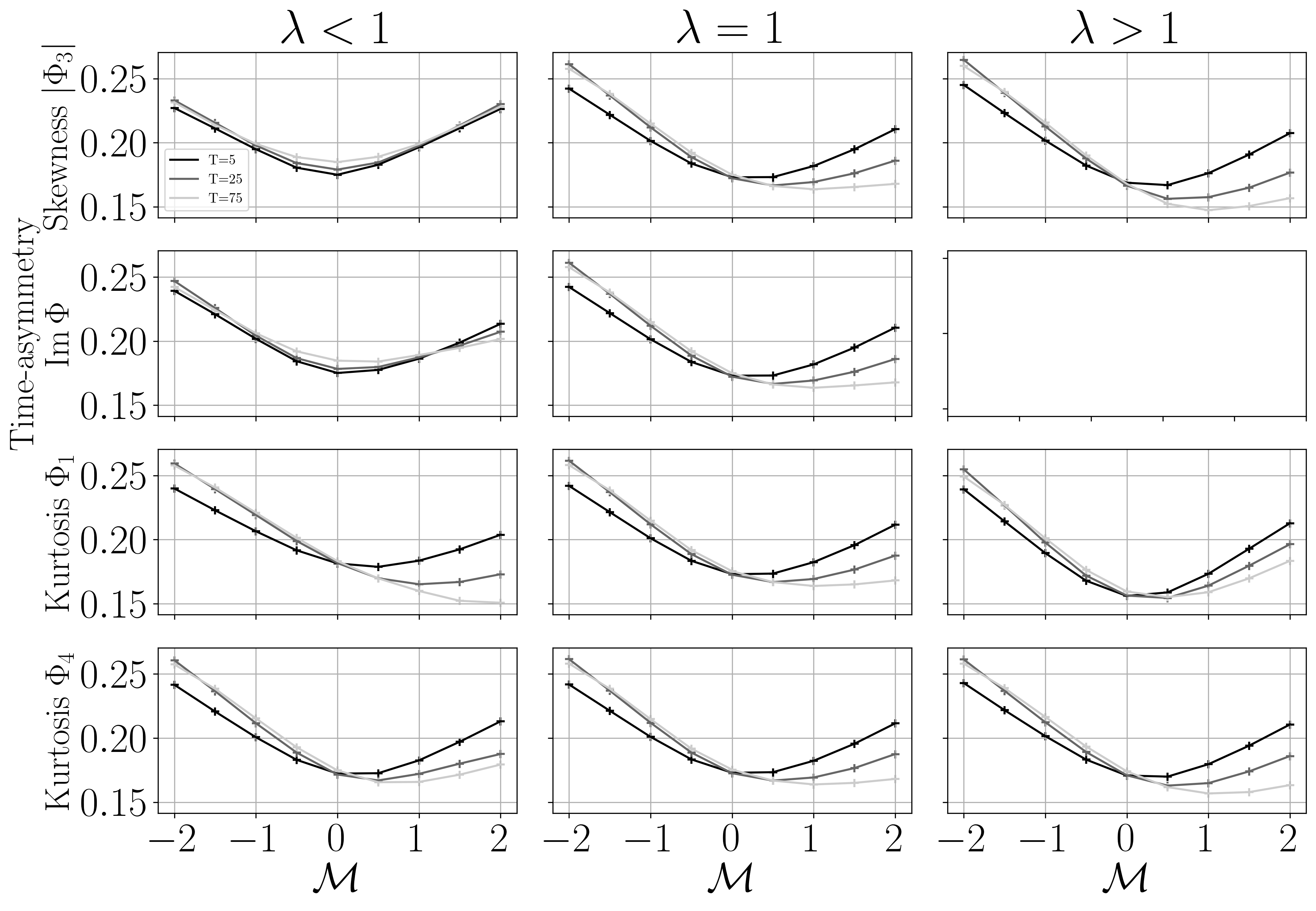}
\caption{Smile sensitivity to change in Scattering Spectra statistics $\Phi(\sple{x})$, decomposed as changes in skewness $|\Phi_3|$, time-asymmetry Im$\,\Phi$, kurtosis $\Phi_1$ or kurtosis $|\Phi_4|$. 
A value $\lambda<1$ indicates respectively, no skewness, no time-asymmetry, less kurtosis. A factor $\lambda=1$ does not change the statistics $\Phi(\sple{x})$ estimated on S\&P.
Besides well-known influence of the skewness and kurtosis on the shape of the smile,  
the Scattering Spectra $\Phi(x)$ also explains the contribution of time-asymmetry in the shape of the smile.}
\label{fig:smile-sensitivity}
\end{figure}

\section{The Path-Dependent Volatility Model}
\label{app:pdv}

The Path-Dependent Volatility (PDV) model introduced in~\cite{guyon2021volatility} consists of a 4-factor Markovian model with $9$ parameters. Writing the price process as $S_t=S_0 e^{x_t}$, it assumes that
\begin{align*}
\frac{dS_t}{S_t} & = \sigma_t dW_t,
\\
\sigma_t & = \sigma(R_{1,t},R_{2,t}),
\\
\sigma(R_1,R_2) & = \beta_0 + \beta_1 R_1 + \beta_2 \sqrt{R_2}
\\
R_{1,t} & = \int_{-\infty}^t K_1(t-u)\frac{{\rm d}S_u}{S_u}
\\
R_{2,t} & = \int_{-\infty}^t K_2(t-u)\big(\frac{{\rm d}S_u}{S_u}\big)^2
\end{align*}
Among these parameters, $6$ are used to parameterize the kernels $K_1,K_2$ both being 
a linear combination of exponentials, and $3$ are the regression coefficients $\beta_0,\beta_1,\beta_2$. 
The $6$ kernel parameters are set to the optimal values calibrated on S\&P and presented in~\cite{guyon2021volatility} (Table 11). 

To obtain unconditional smiles in the PDV model, the process is evolved with $1$ step per day until it reaches a stationary regime.
The $3$ regression coefficients $\beta_0,\beta_1,\beta_2$ are calibrated on the Scattering Spectra statistics $\Phi$ and reported in Table~\ref{tab:pdv-unconditional-calib}. We choose the values that minimize the distance $\|\Phi(x_\text{PDV}) - \Phi(\sple{x})\|_2$ where $\sple{x}$ is the S\&P log-price time series consisting of $N=5827$ days from January 2000 to April 2023 and $x_\text{PDV}$ are 10 realizations of same size from the PDV model. 
The Scattering Spectra of the `closest' PDV model to the S\&P are shown on Fig.\ref{fig:scat-spectra}.

\begin{table}[htbp]
\centering
\begin{subtable}{0.45\textwidth}
    \centering
    \begin{tabular}{|c|c|c|c|}
        \hline
        & $\beta_0$ & $\beta_1$ & $\beta_2$
        \\
        \hline
        \hline
        unconditional & 0.050 & -0.13 & 0.56
        \\
        \hline
    \end{tabular}
    \caption{Optimal parameters of the Path-Dependent Volatility model calibrated on Scattering Spectra statistics over the full period 2000-2023.
    }
    \label{tab:pdv-unconditional-calib}
\end{subtable}
\hfill
\begin{subtable}{0.45\textwidth}
    \centering
    \begin{tabular}{|c|c|c|c|}
        \hline
        & $\beta_0(T)$ & $\beta_1(T)$ & $\beta_2(T)$
        \\
        \hline
        \hline
        $T=1$ & 0.022 & -0.062 & 0.64
        \\
        \hline
        $T=7$ & 0.035 & -0.066 & 0.75
        \\
        \hline
        $T=8$ & 0.036 & -0.065 & 0.74 
        \\
        \hline
        $T=25$ & 0.051 & -0.050 & 0.68
        \\
        \hline
        $T=50$ & 0.066 & -0.040 & 0.62
        \\
        \hline
        $T=75$ & 0.079 & -0.039 & 0.55
        \\
        \hline
        $T=126$ & 0.098 & -0.033 & 0.46
        \\
        \hline
        $T=150$ & 0.10 & -0.030 & 0.43
        \\
        \hline
    \end{tabular}
    \caption{Optimal parameters of the Path-Dependent Volatility model on the period 2000-2014 (train period), obtained through linear regression of future realized volatilities over $T$ days.}
    \label{tab:pdv-conditional-calib}
\end{subtable}
\caption{Path-Dependent Volatility (PDV) model calibration.}
\end{table}

It shows that the kurtosis $\Phi_4(x)$ is significantly lower than for the S\&P500 data, see also Fig.~\ref{fig:atm-kurtosis}.
Looking at the log-return trajectories shown in Fig.~\ref{fig:stat-pdv} we indeed notice clear qualitative discrepancies: the PDV log-returns are less intermittent than the S\&P ones. 
In Fig.~\ref{fig:stat-pdv} we see that structure functions (b) are poorly reproduced. Furthermore, the short term leverage effect is much too strong, see and, quite strikingly, Fig.~\ref{fig:pdv-leverage} and Fig.~\ref{fig:atm-skew}.

We could also choose parameters $\beta_0,\beta_1,\beta_2$ as the one specified in~\cite{guyon2021volatility}, using optimal parameters for the regression of future realized volatility but it does not improve the results in Figs.~\ref{fig:stat-pdv},\ref{fig:avg-smile}.
In the former case, the PDV trajectories feature a larger skewness $\Phi_3$ than the S\&P with abnormal negative log-return values, as shown in~\cite{guyon2021volatility}.
In the later case, the intermittency, measured in particular by $\Phi_4$, is much lower than the S\&P.

This showcases the difficulty for a low-parametric model to capture a rich set of statistics such as the Scattering Spectra. 

To obtain good volatility prediction and conditional smiles across different maturities $T$, to be used in trading games (see section \ref{sec:option-pricing}), we had to recalibrate the parameters $\beta_0(T),\beta_1(T),\beta_2(T)$ for each maturity $T$ independently, this in order to provide the best prediction of the future realized variance i.e. the overall level of the smile.
Calibrated values are shown on Table~\ref{tab:pdv-conditional-calib}.
\begin{figure}
\centering 
\includegraphics[width=\linewidth]{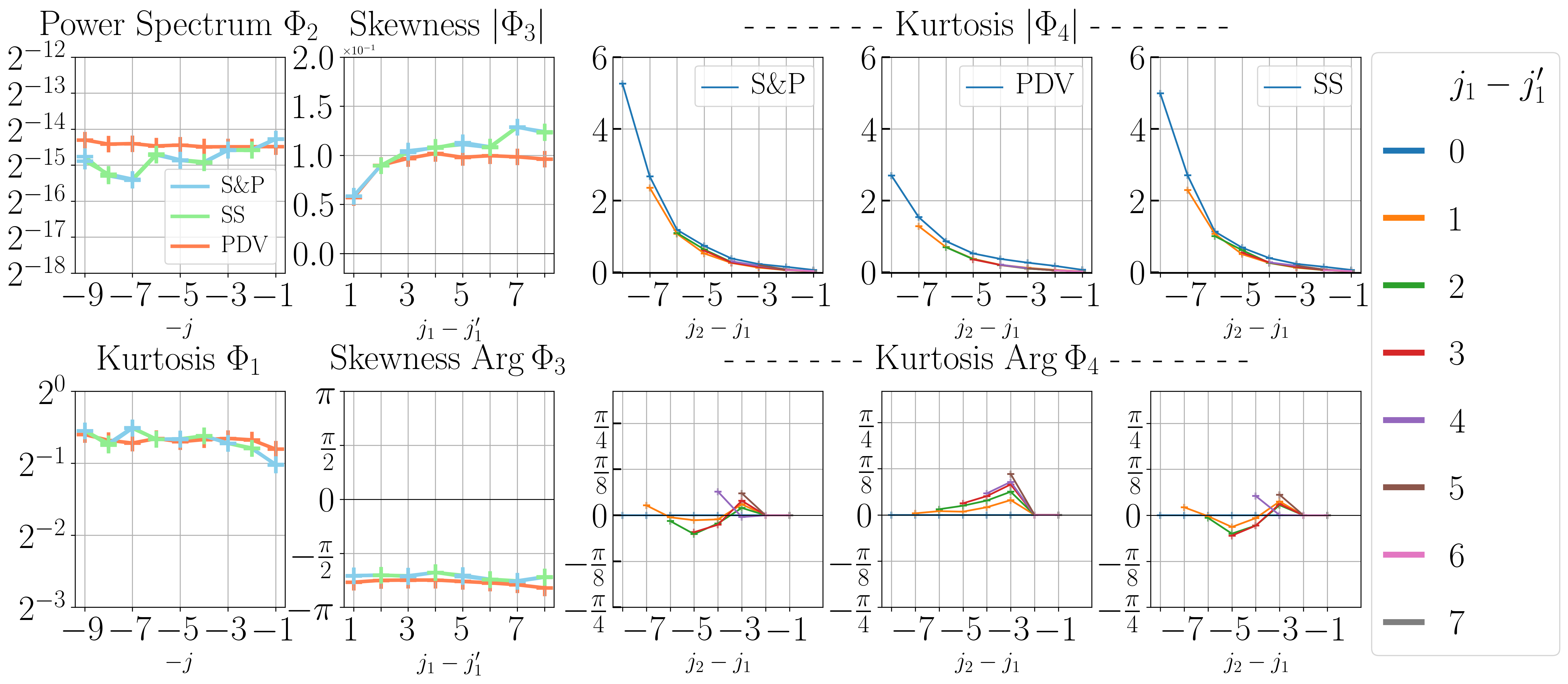}
\caption{Visualization of the Scattering Spectra statistics $\Phi=(\Phi_1,\Phi_2,\Phi_3,\Phi_4)$.
These are estimated on the observed S\&P500 realization, in a Path-Dependent Volatility (PDV) model and in the Scattering Spectra (SS) model.
}
\label{fig:scat-spectra}
\end{figure}
\begin{figure}
\centering
\begin{subfigure}[b]{\linewidth}
    \centering
    \includegraphics[width=\textwidth]{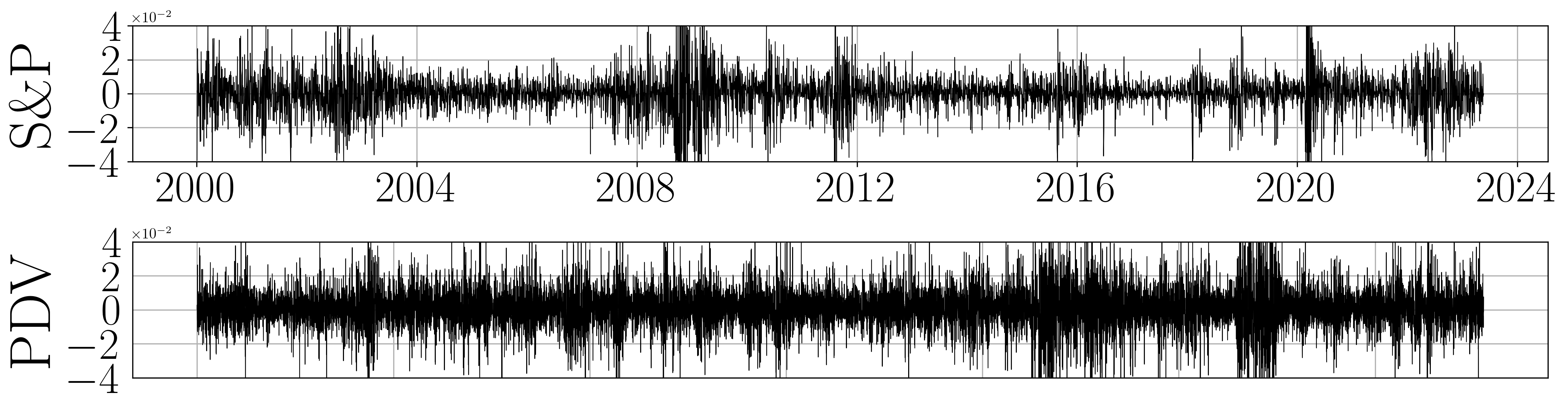}
\end{subfigure}
\vskip\baselineskip
\begin{subfigure}[b]{0.32\linewidth}
    \centering
    \includegraphics[width=\textwidth]{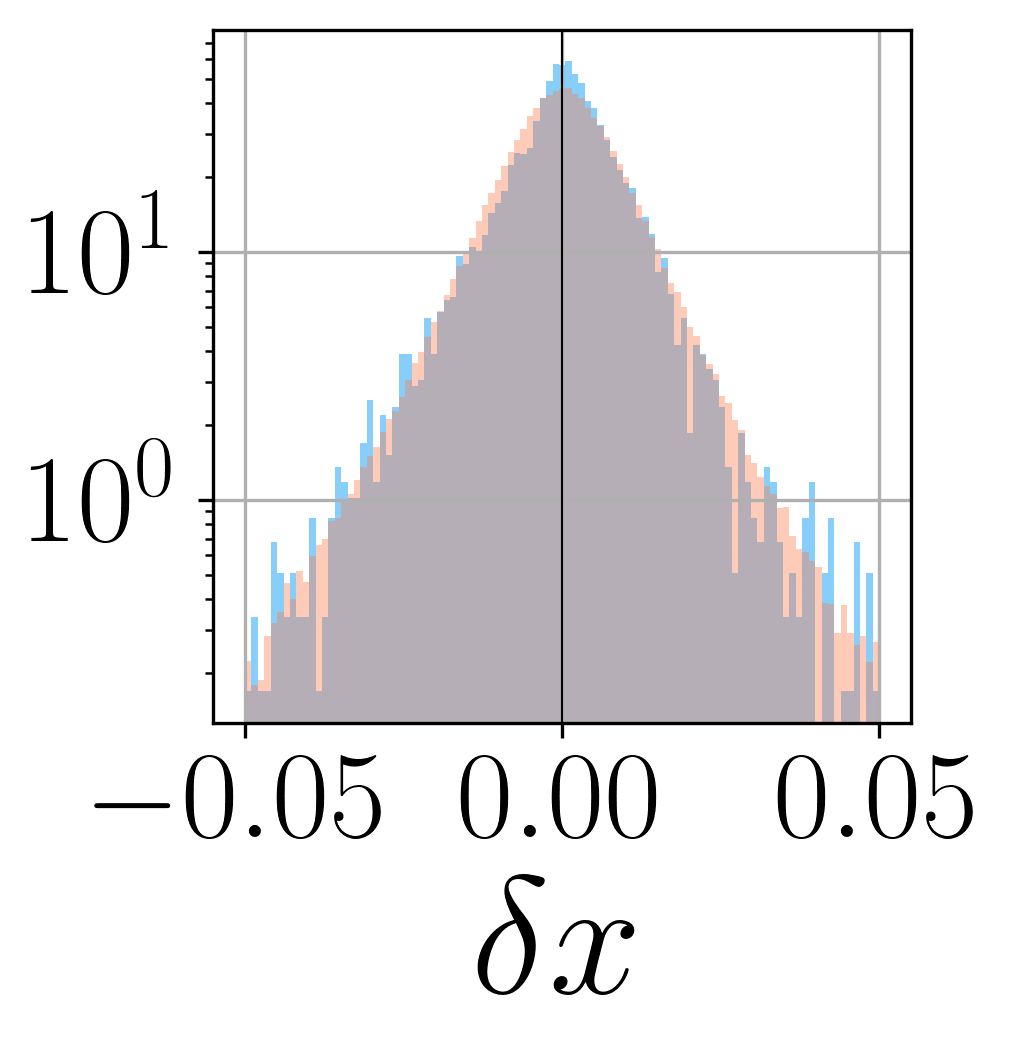}
    \caption{Histogram}
\end{subfigure}
\begin{subfigure}[b]{0.32\linewidth}  
    \centering 
    \includegraphics[width=\textwidth]{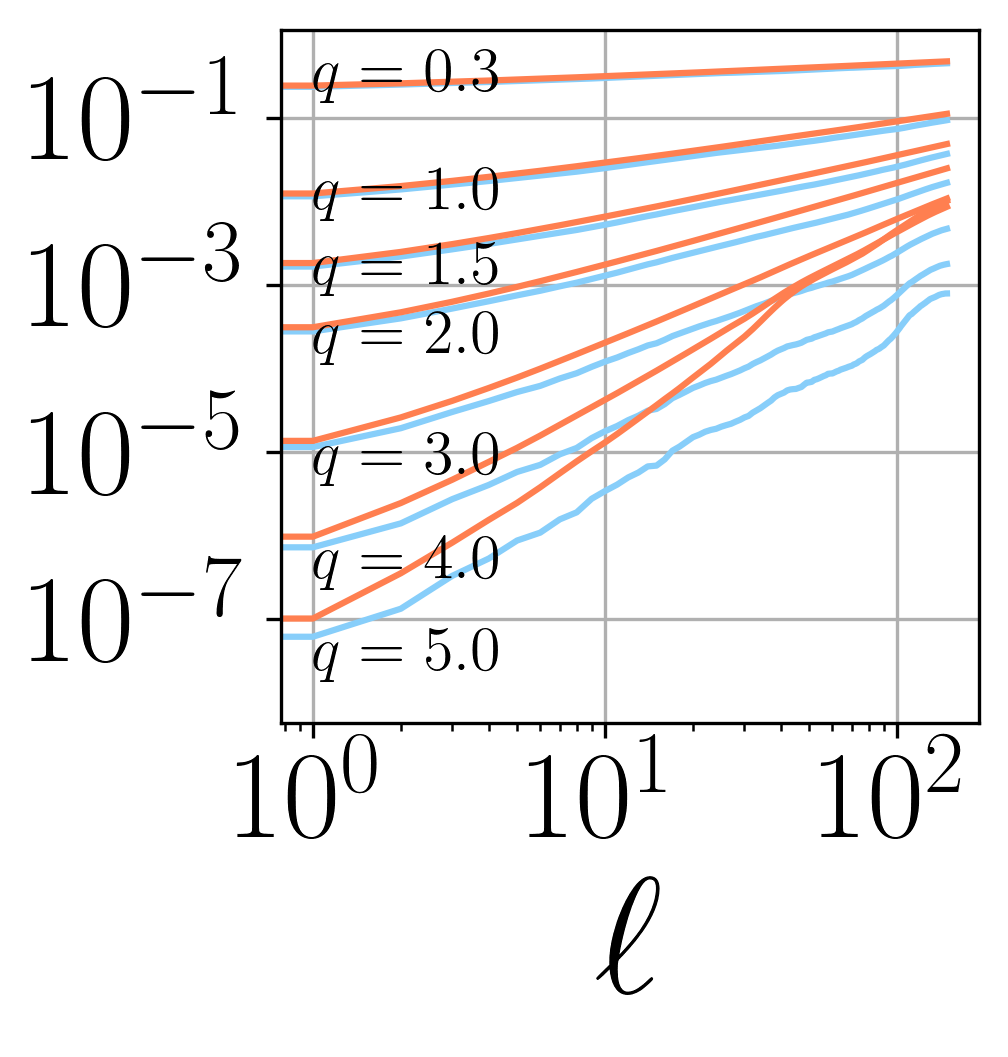}
    \caption{Struct. functions}
\end{subfigure}
\begin{subfigure}[b]{0.32\linewidth}   
    \centering 
    \includegraphics[width=\textwidth]{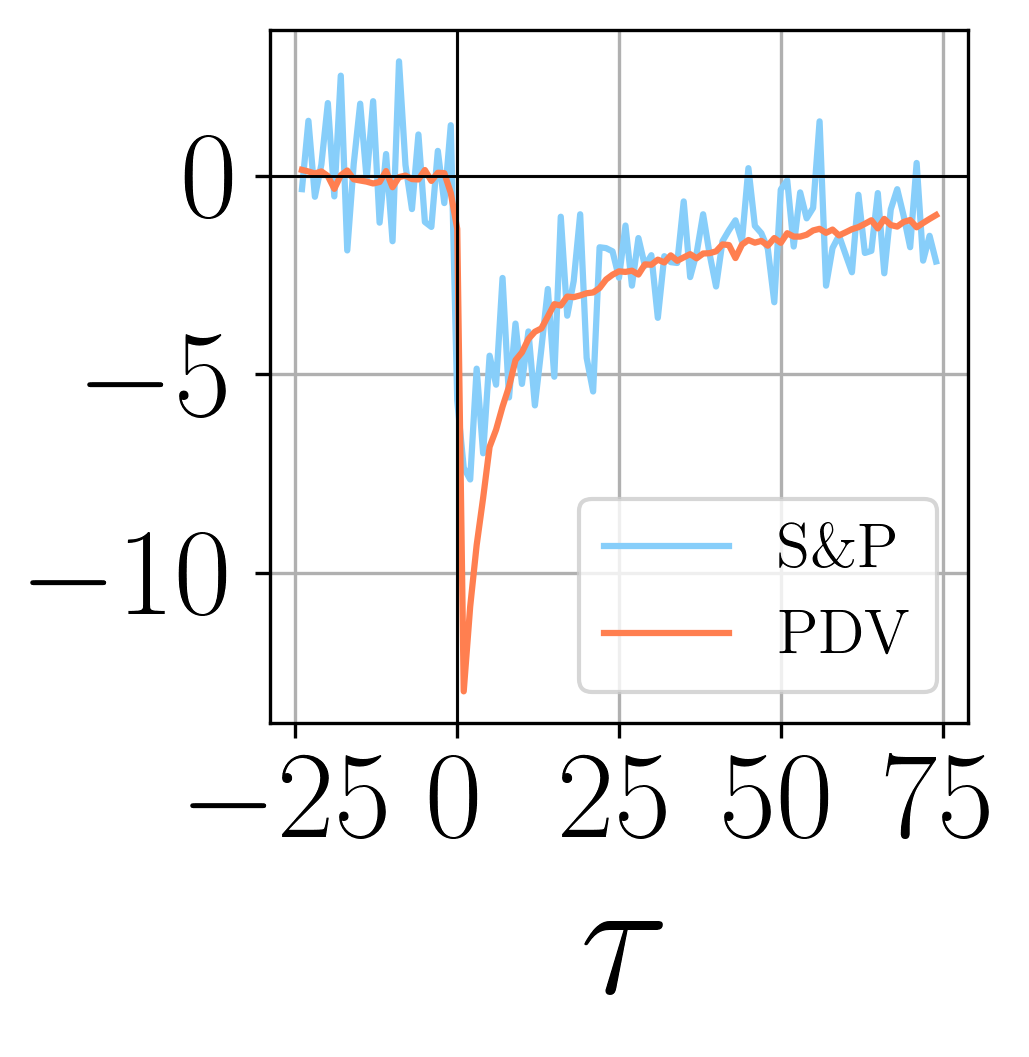}
    \caption{Leverage}
    \label{fig:pdv-leverage}
\end{subfigure}
\caption{S\&P returns vs. synthetic returns using the Path-Dependent Volatility (PDV) model calibrated on Scattering Spectra (top), and standard statistics in the same PDV model (orange) of the S\&P (blue):
(a) Histogram of daily log-returns $\delta x$. (b) Structure functions $\E\{|\delta_\ell x(t)|^q\}$. (c) Leverage correlation $\E\{\delta x(t-\tau)|\delta x(t)|^2\}$. Compare with the results shown in Fig. ~\ref{fig:stat-our-model}.}
\label{fig:stat-pdv}
\end{figure}

\section{Proof of theorem \ref{theo:shadowing-MC}}
\label{app:proofs}

Let us write $w_k=c_n g_\eta(x^k-\sple{x})$ where 
\[
g_\eta(x)= (\eta \sqrt{2\pi})^{-N'} e^{-\frac1{2\eta^2}\|\pa{x}\|^2}
\]
is a Gaussian kernel with $N'$ being the dimension of $\pa{x}$ and $c_n$ is such that $\frac1n\sum_k w_k = 1$.
We write $\hat q$ the estimator 
\[
\hat q = \frac{1}{n}\sum_{k=1}^n w_k q(x^k).
\]
We prove the convergence of $\hat q$ to $\E\{q(x) ~|~ \pa{x} = \pa{\sple{x}}\}$ almost surely by first taking the limit $n\to+\infty$ and then $\eta\to0$.

\noindent\textbf{Limit $n\to\infty$.} One can calculate $c_n^{-1}=\frac1n\sum g_\eta(x^k)$ so that one has
\[
\hat q = \frac{\frac1n \sum_{k=1}^n g_\eta(x^k - \sple{x}) q(x^k)}{\frac1n \sum_{k=1}^n g_\eta(x^k-\sple{x})}.
\]
Since $g_\eta$ is bounded one has $\E\{g_\eta(x-\sple{x}) |q(x)|\} < +\infty$ and $\E\{g_\eta(x-\sple{x})\}< +\infty$. From the law of large numbers, knowing that $\E\{g_\eta(x-\sple{x})\} > 0$, it follows
\[
\hat q \underset{n\to+\infty}{\longrightarrow} \frac{\E\{g_\eta(x-\sple{x})q(x)\}}{\E\{g_\eta(x-\sple{x})\}}.
\]

\noindent\textbf{Limit $\eta\to0$.}
We will make use of the following lemma of approximation by convolution, proved for example in~\cite{evans2022partial}.
\begin{lemma}
If $f\in\mathcal{C}^0\cap L^1(dx)$ then for all $\sple{x}\in\R^N$
\[
g_\eta\star f(\sple{x}) \underset{\eta\to0}{\longrightarrow} \int f(\pa{\sple{x}}, \fu{x}) {\rm d}\fu{x}.
\]
\end{lemma}
Let us notice that $\E\{g_\eta(x-\sple{x})\} = \int g_\eta(\sple{x}-x) p(x) {\rm d}x = g_\eta\star p (\sple{x})$, and $\E\{g_\eta(x-\sple{x})q(x)\} = g_\eta\star (qp) (\sple{x})$. 
Since $\E\{q(x)\}<+\infty$ one has $qp\in L^1(dx)$, $p$ being a probability distribution one also has $p\in L^1(dx)$.
From the lemma we get:
\[
\frac{\E\{g_\eta(x-\sple{x})q(x)\}}{\E\{g_\eta(x-\sple{x})\}} 
\underset{\eta\to0}{\longrightarrow} 
\frac{\int q(\pa{\sple{x}},\fu{x})p(\pa{\sple{x}},\fu{x}) {\rm d}\fu{x}}{\int p(\pa{\sple{x}},\fu{x}) {\rm d}\fu{x}} ,
\]
where the denominator is non-zero because $p(x)>0$ for all $x\in\R^N$.
The former term being $\E\{q(x) ~|~ \pa{x}=\pa{\sple{x}}\}$, this proves the theorem.

\section{Choice of past embedding $h$}
\label{app:choice-of-h}

The following proposition shows that the choice of $h$ in our paper (\ref{eq:choice-of-h}) induces equivariance properties on the set of shadowing paths $H_\eta(\pa{\sple{x}})$. We recall that we chose $\eta=\widehat\eta\|h(\pa{\sple{x}})\|$ for a fixed $\widehat\eta$.
These equivariance properties are proved on continuously sampled paths $\pa{x}=(x(t), t<0)$ with $t\in\R$. In that case, we still write $h_{\alpha,\beta}$ the continuously sampled analogue 
\[
h_{\alpha,\beta}(x) = \Big( \frac{ x(0) - x(-\ell)}{\ell^\beta} 
~ , ~ \ell=\alpha^m  ~ , ~ m\in\Z \Big)
\]
which is now of infinite dimension.
\begin{proposition}
For $h=h_{\alpha,\beta}$ with $\alpha>0$ and $\beta>1$ one has
\\
1. (Multiplication equivariance) for $\lambda > 0$
\[
H_\eta(\lambda\,\pa{\sple{x}}) = \lambda.\sha
\]
\\
2. (Dilation equivariance) writing $\Gamma x(t) = x(\alpha t)$
\[
H_\eta(\Gamma\pa{\sple{x}}) = \Gamma.\sha
\]
\end{proposition}
\begin{proof}[Proof]
The first equivariance follows directly from the fact that $h_{\alpha,\beta}$ is itself equivariant to multiplication
$h_{\alpha,\beta}(\lambda \pa{x}) = \lambda\,h_{\alpha,\beta}(\pa{x})$.
For dilation, one has 
\begin{align*}
 h_{\alpha,\beta}(\Gamma\pa{x}) & = \Big( \frac{ x(0) - x(-\alpha\ell)}{\ell^\beta} 
~ , ~ \ell=\alpha^m  ~ , ~ m\in\Z \Big)
 \\
 & = \alpha^\beta \Big( \frac{ x(0) - x(-\ell)}{\ell^\beta} 
~ , ~ \ell=\alpha^{m+1}  ~ , ~ m\in\Z \Big)
\end{align*}
This means that $h_{\alpha,\beta}(\Gamma\pa{x})$ is equal to $h_{\alpha,\beta}(\pa{x})$ up to a shift in indices and up to a multiplicative constant.
It follows that 
\[
\| h_{\alpha,\beta}(\Gamma\pa{x}) -  h_{\alpha,\beta}(\Gamma\pa{\sple{x}})\| = \alpha^{\beta}\| h_{\alpha,\beta}(\pa{x}) -  h_{\alpha,\beta}(\pa{\sple{x}})\|.
\]
Now, normalizing by $\|h_{\alpha,\beta}(\Gamma \pa{\sple{x}})\| = \alpha^\beta \|h_{\alpha,\beta}( \pa{\sple{x}})\|$ yields $H_\eta(\Gamma\pa{\sple{x}}) = \Gamma.\sha$.
\end{proof}

\section{Additional trading game statistics}
\label{app:trading-game}

In addition to the P\&Ls and aggregated P\&Ls of a trading game played against the option market shown in Figs.~\ref{fig:smile-prediction-avg},\ref{fig:smile-prediction-aggregated}, we show in Figs.~\ref{fig:smile-prediction-std},\ref{fig:smile-prediction-win} the standard deviation and winning rate, defined as the average number of times
the payoff of a trade is positive.

In the remaining figures, we also show the statistical results of the trading game between PDV and SS model. These results do not require option market data (but require actual price series of the underlying) and directly test the relative quality of purely statistical price models. As seen in Fig.~\ref{fig:pdv_ssb_game}, trading game  unequivocally favours the Scattering Spectra model framework over the Path-Dependent Volatility model.
\begin{figure}
\centering
\includegraphics[width=0.9\linewidth]{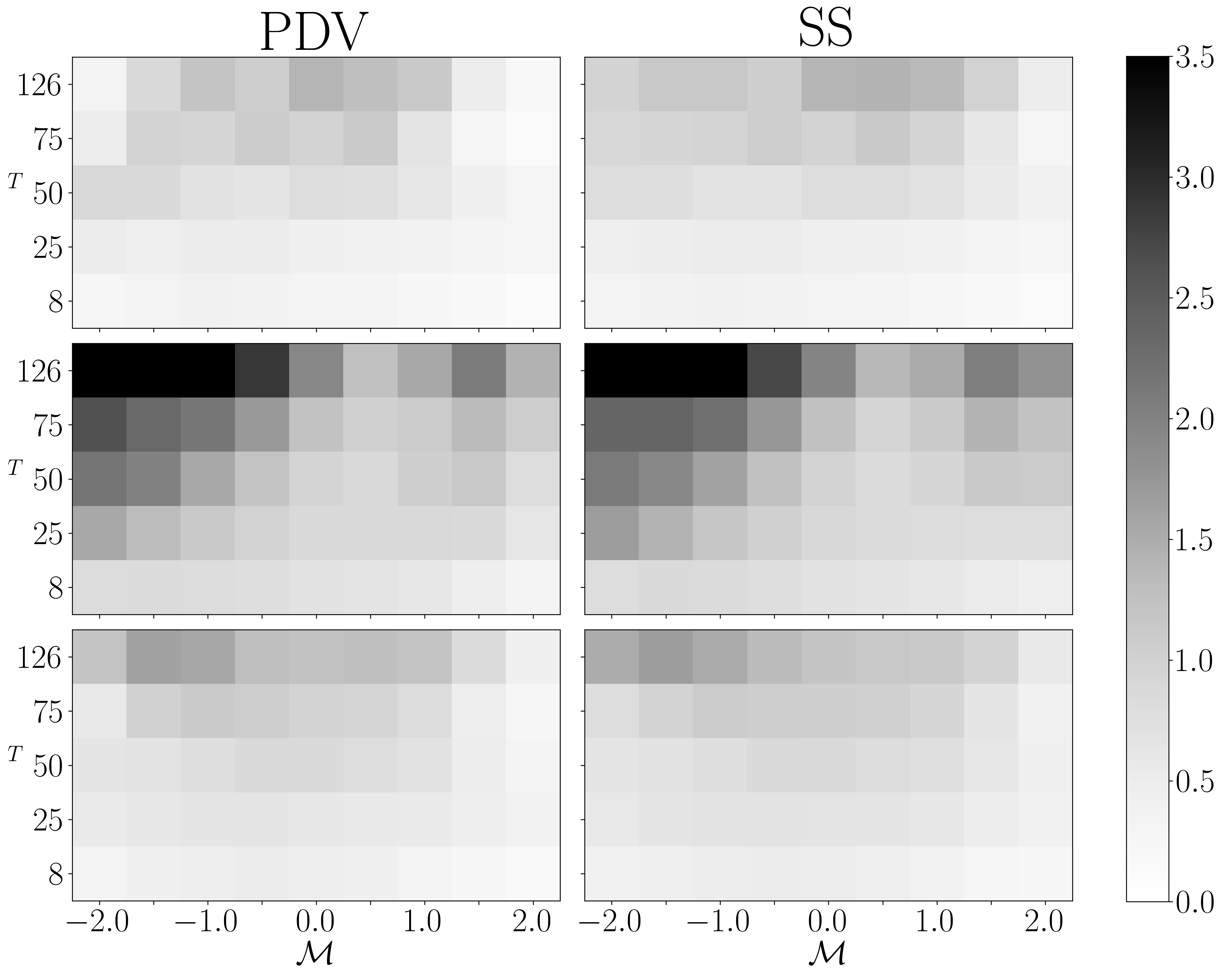}
\caption{Standard deviation of $\mathrm{P\&Ls}$ of a trading game playing the Scattering Spectra (SS) model vs. S\&P or the Path-Dependent Volatility (PDV) model vs. S\&P.
Each heatmap corresponds to a 3 years period, from top to bottom (2015-2017, 2018-2020, 2021-2023). }
\label{fig:smile-prediction-std}
\end{figure}
\begin{figure}
\centering
\includegraphics[width=0.9\linewidth]{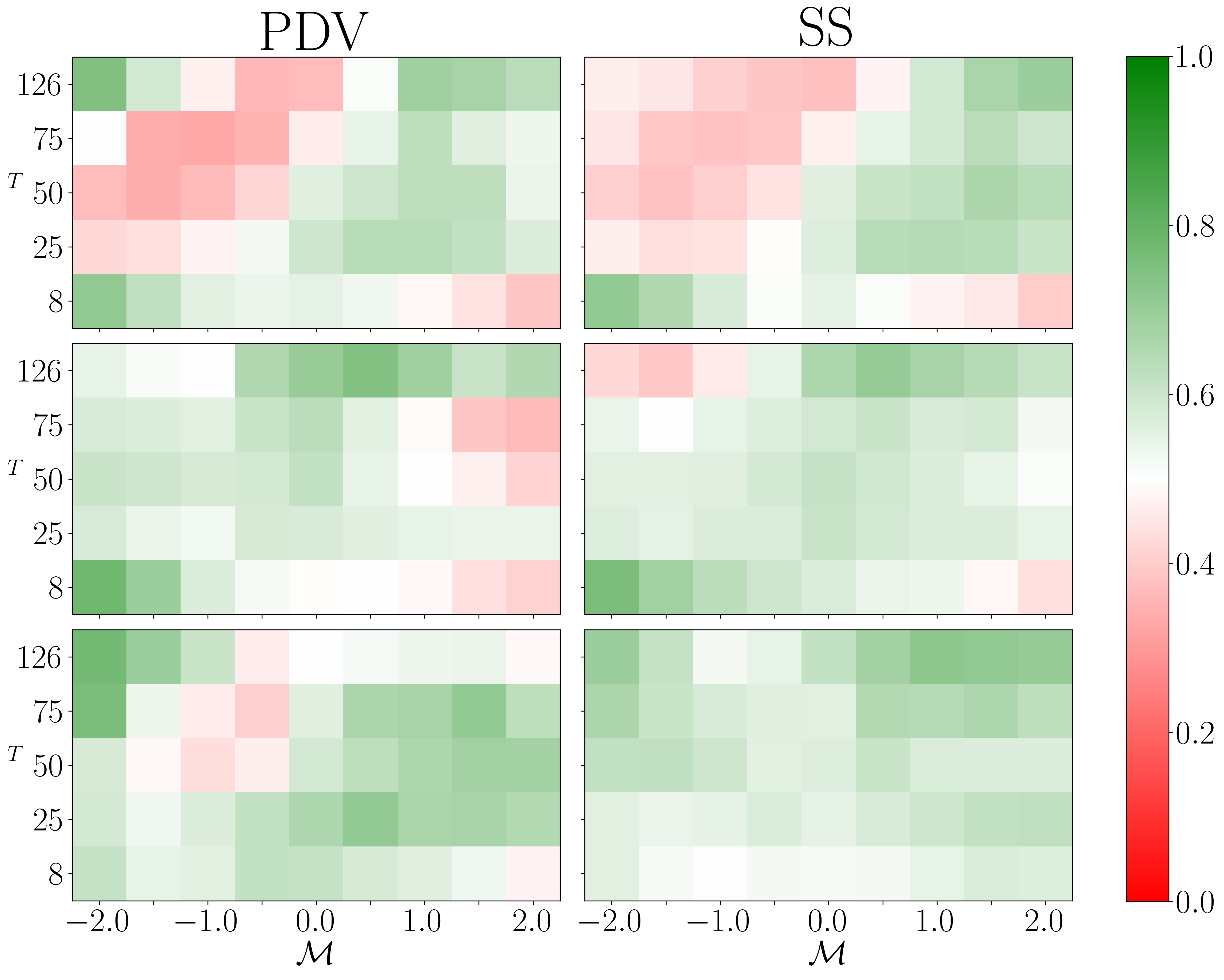} 
\caption{Rate of winning trades of a trading game playing the Scattering Spectra (SS) model vs. S\&P or the Path-Dependent Volatility (PDV) model vs. S\&P.
Each heatmap corresponds to a 3 years period, from top to bottom (2015-2017, 2018-2020, 2021-2023).}
\label{fig:smile-prediction-win}
\end{figure}
\begin{figure}
\centering
\includegraphics[width=1\linewidth]{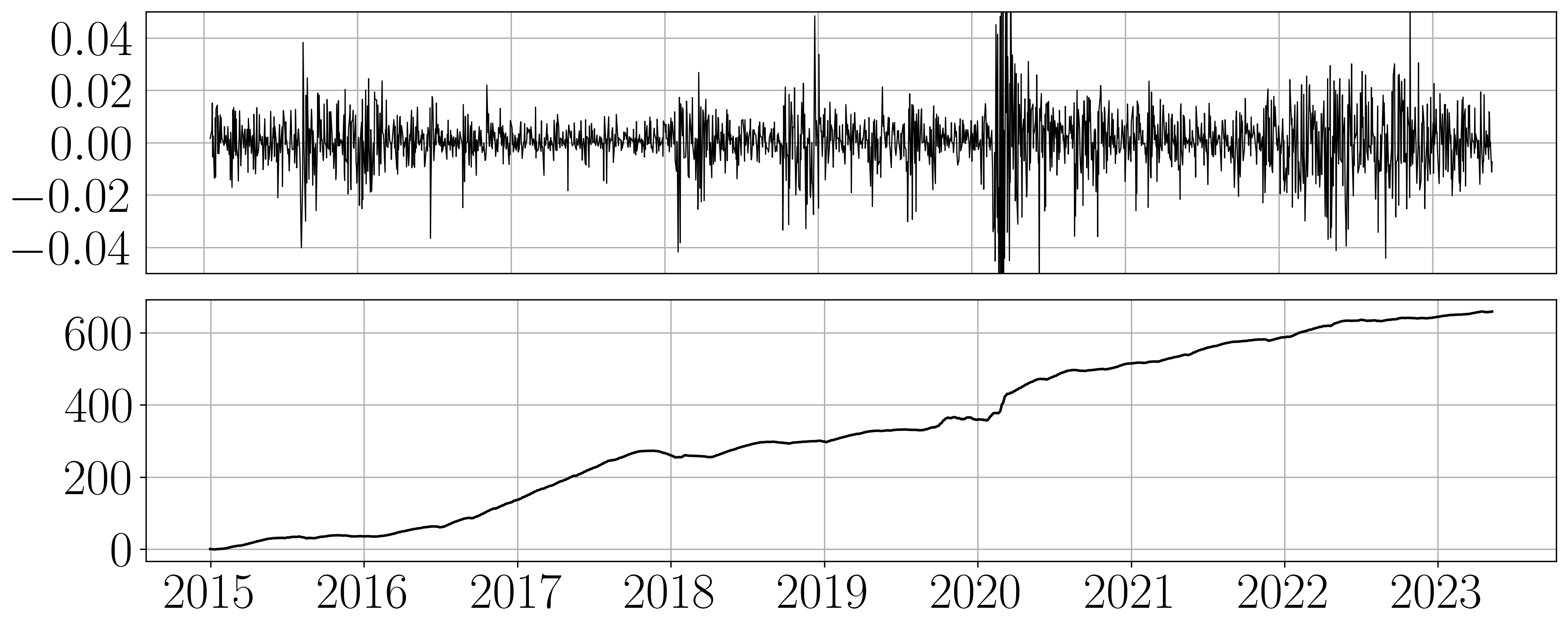}
\caption{
(Bottom) Aggregated $\mathrm{P\&L}$ of a trading game playing the Scattering Spectra model vs. the Path-Dependent Volatility model (re-optimized for each $T$), from 2015 to 2023, both calibrated using data from 2000 to 2014.
(Top) Time series of S\&P log-return on the trading game period.}
\label{fig:pdv_ssb_game}
\end{figure}
\begin{figure}
\centering
\includegraphics[width=0.7\linewidth]{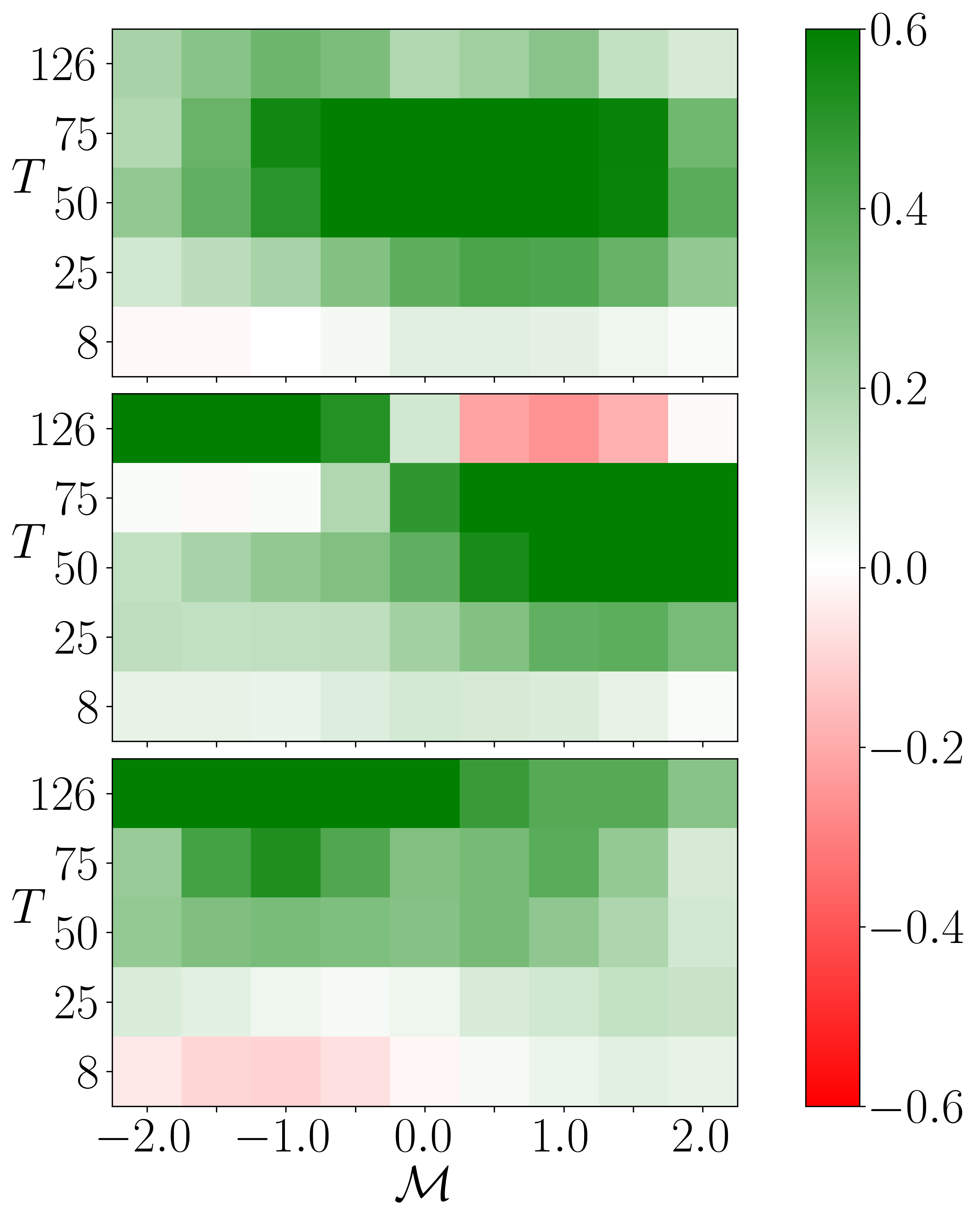} 
\caption{$\mathrm{P\&Ls}$ in a trading game playing the Scattering Spectra model against the Path-Dependent Volatility model.
Each heatmap corresponds to a 3 years period, from top to bottom (2015-2017, 2018-2020, 2021-2023).}
\end{figure}
\begin{figure}
\centering
\includegraphics[width=0.7\linewidth]{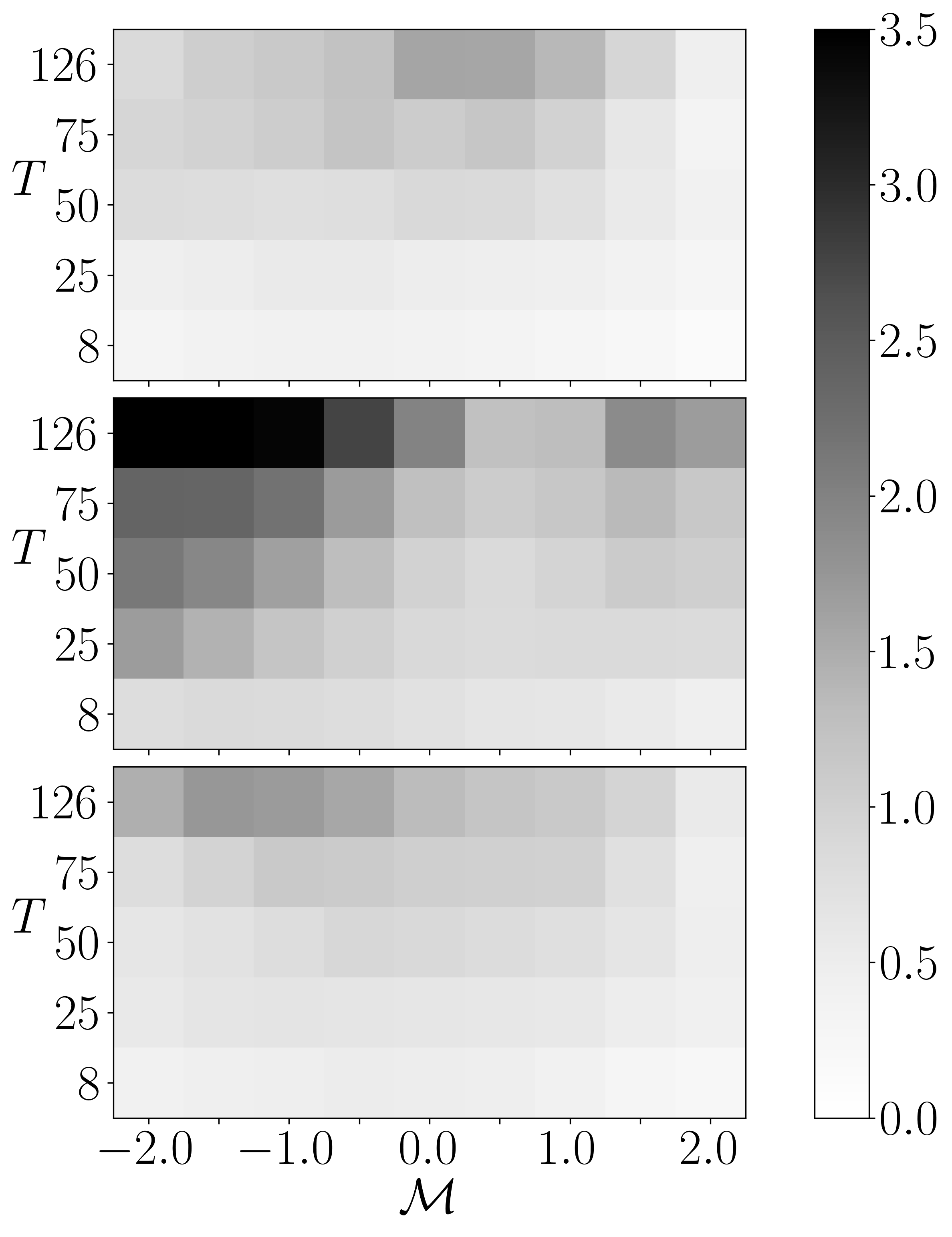} 
\caption{Standard deviation of $\mathrm{P\&Ls}$ in a trading game playing the Scattering Spectra model against the Path-Dependent Volatility (PDV) model.
Each heatmap corresponds to a 3 years period, from top to bottom (2015-2017, 2018-2020, 2021-2023).}
\end{figure}
\begin{figure}
\centering
\includegraphics[width=0.7\linewidth]{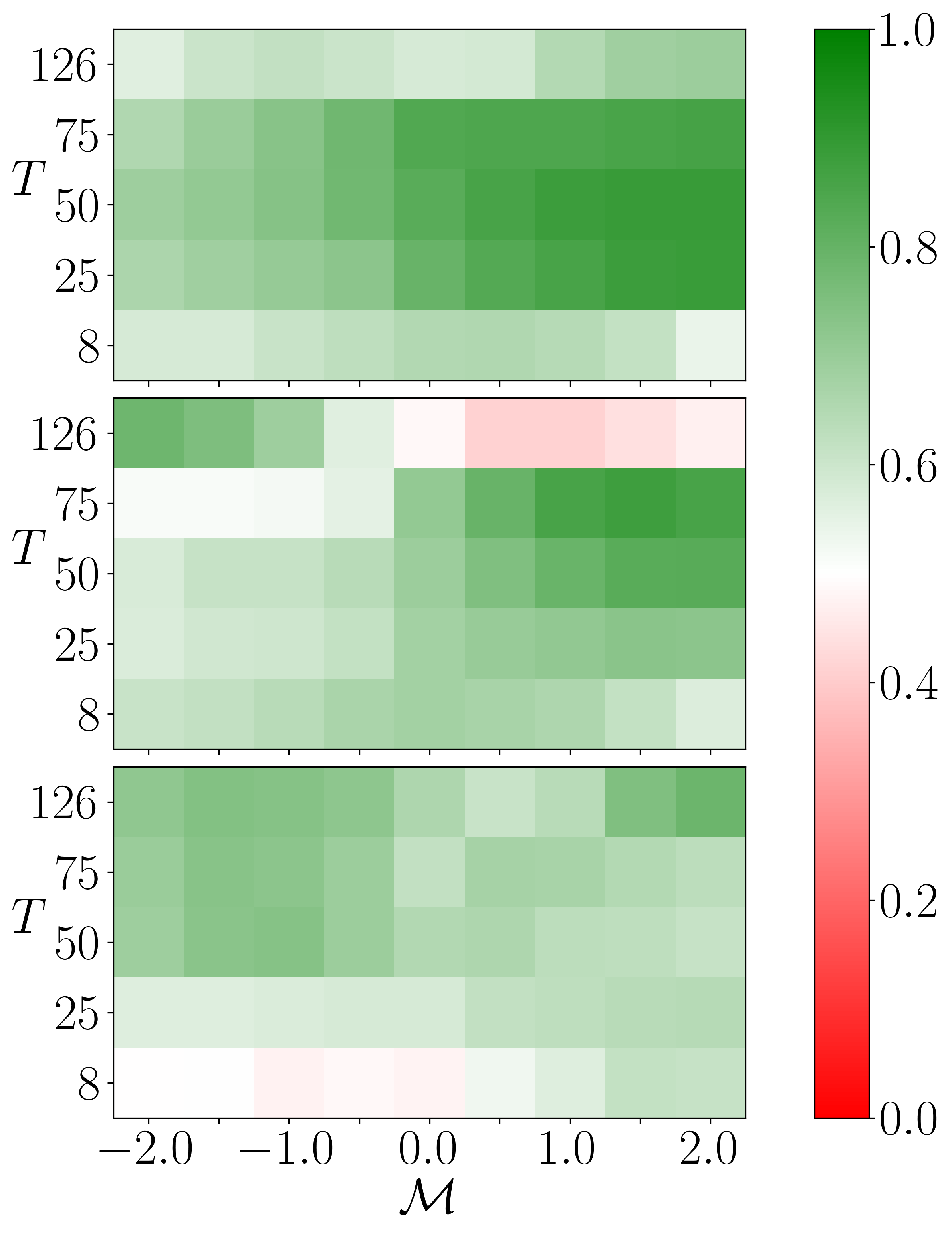} 
\caption{Percentage of winning trades of a trading game playing the Scattering Spectra model against Path-Dependent Volatility model. 
Each heatmap corresponds to a 3 years period, from top to bottom (2015-2017, 2018-2020, 2021-2023).}
\end{figure}

\end{document}